\documentclass[oribibl, envcountsame]{llncs}

\newif\ifextended\extendedtrue

\usepackage{amsmath}
\usepackage{amssymb}
\usepackage{stmaryrd}
\usepackage{subfigure}
\usepackage{relsize}
\usepackage{graphicx}
\usepackage{mathrsfs}
\usepackage{ifthen}
\usepackage{url}
\usepackage{bbm}
\usepackage{tikz}
\usetikzlibrary{backgrounds}
\usetikzlibrary{automata}
\usetikzlibrary{arrows}
\usetikzlibrary{shapes}
\usetikzlibrary{decorations.pathmorphing}
\usetikzlibrary{fit}
\usetikzlibrary{calc}
\usepackage{wrapfig}

\newcommand{\ctauequiv}{\stackrel{\ctau\,}{\rightsquigarrow}}

\makeatletter 
\providecommand\dotsum{\mathpalette\@dotted\sum \vphantom{\sum}} 
\def\@dotted#1#2{\ooalign{\hfil$#1 \bullet $\hfil\cr\hfil$#1#2$\hfil\cr}} 
\makeatother 
\DeclareMathOperator*{\psum}{\mathchoice{\sum\!\!\!\!\!\!\!\bullet\ \ }{\sum\hspace{-2.15ex}\raisebox{0.2ex}{\scalebox{0.75}{$\bullet$}}\,\,}{\sum\!\!\!\!\!\!\!\bullet}{\sum\!\!\!\!\!\!\!\bullet}}

\DeclareMathAlphabet{\mathitbf}{OML}{cmm}{b}{it}

\let\vec\mathitbf

\newcommand{\bottom}{{\mathrel{\perp}}}

\newcommand{\dirac}[1]{\mathbbmss{1}_{#1}}

\tikzstyle{every scope}=[]
\tikzstyle{every picture}=[thick]
\tikzstyle{every loop}=[->]

\pgfarrowsdeclaredouble[-0.03cm]{todouble}{todouble}{to}{to}%

\tikzstyle{nodesmall} = [ minimum size=2mm, inner sep=0.5mm]
\tikzstyle{onEdge}=[fill=white, pos=0.5]
\tikzstyle{loopupper}=[in=67, out=113, loop]
\tikzstyle{lagerLabel}=[pos=0.68]
\tikzstyle{prob}=[->, decorate, decoration={snake,segment length=2mm, amplitude=.4mm,post length=1mm}]

\tikzstyle{loopleft}=[in=-176, out=-184, loop, swap]
\tikzstyle{loopright}=[in=-10, out=10, loop, swap]
\tikzstyle{loopleft2}=[in=-171, out=-189, loop, swap]

\newlength{\rlength}
\newlength{\ulength}
\newlength{\arclength}
\newlength{\labellengthin}
\newlength{\labellengthout}
\pgfmathadd{0}{1}                         
\pgfmathparse{\pgfmathresult * tan(atan(0.5)/2)}
\def\clap#1{\hbox to 0pt{\hss#1\hss}}
\def\mathclap{\mathpalette\mathclapinternal}

\def\mathclapinternal#1#2{%
            \clap{$\mathsurround=0pt#1{#2}$}}

\newcommand{\en}{\mathrel{\wedge}}
\newcommand{\of}{\mathrel{\vee}}
\newcommand{\comp}{\mathrel{\circ}}
\newcommand{\union}{\mathrel{\cup}}

\newcommand{\qdot}{\; . \;}

\newcommand{\ctau}{\tau_\textrm{\rm c}}
\newcommand{\trans}{{\Delta}}

\newcommand{\finpaths}{\textit{finpaths}}
\newcommand{\prefix}{\textit{prefix}}
\newcommand{\maxpaths}{\textit{maxpaths}}
\newcommand{\trace}{\textit{trace}}
\newcommand{\last}{\textit{last}}
\newcommand{\finpi}{\pi}
\newcommand{\infpi}{\pi'}
\newcommand{\transition}{\textit{step}}
\newcommand{\schedprob}{P}
\newcommand{\distrprob}{F}

\newcommand{\pa}{\mathcal{A}}
\newcommand{\sched}{\mathcal{S}}
\DeclareMathOperator{\distr}{Distr}
\DeclareMathOperator{\distrs}{Distr*}

\newcommand{\srightarrow}{\vphantom{\raisebox{0.05cm}{.}}\smash{\rightarrow}}
\newcommand{\sRightarrow}{\vphantom{\raisebox{0.08cm}{.}}\smash{\Longrightarrow}}
\newcommand{\scurlyrightarrow}{\vphantom{\raisebox{0.079cm}{.}}\smash{\leadsto}}
\newcommand{\bb}{\mathrel{\leftrightarroweq_\textrm{\rm bp}}}
\newcommand{\curlystep}[1]{\mathrel{\overset{\!#1}{\scurlyrightarrow}}}
\newcommand{\weakstep}[1]{\mathrel{\overset{\!#1}{\sRightarrow}}}
\newcommand{\tsrc}[1]{\mathrel{\overset{#1}{\vphantom{\srightarrow}\smash{\twoheadleftarrow\!\twoheadrightarrow}}}}

\newcommand{\sfrac}[2]{$\smash{\frac{#1}{#2}}$}

\newdimen\arrowx
\newdimen\arrowy
\newdimen\secondarrowx
\newdimen\secondarrowy
\newdimen\rootx    
\newdimen\rooty
           
\newcommand{\drawarc}[5]{
	\pgfextractx{\arrowx}{\pgfpointanchor{#2}{center}};
      	\pgfextracty{\arrowy}{\pgfpointanchor{#2}{center}};
        	\pgfextractx{\secondarrowx}{\pgfpointanchor{#3}{center}};
        	\pgfextracty{\secondarrowy}{\pgfpointanchor{#3}{center}};
  	\pgfextractx{\rootx}{\pgfpointanchor{#1}{center}};
        	\pgfextracty{\rooty}{\pgfpointanchor{#1}{center}};
 
	\pgfmathsetmacro{\firstAngle}{270 - atan((\the\arrowx-\the\rootx) / (\the\arrowy - \the\rooty) )};
	\pgfmathsetmacro{\secondAngle}{270 - atan((\the\secondarrowx-\the\rootx) / (\the\secondarrowy - \the\rooty) )};
	\pgfmathsetmacro{\xshift}{cos(\firstAngle)};
	\pgfmathsetmacro{\yshift}{sin(\firstAngle)};	 
	
	\draw (#1) + (\xshift\arclength,  \yshift\arclength) arc ( \firstAngle :  \secondAngle : \arclength);
       
       	\ifthenelse{\boolean{#5}}{
       		\pgfmathsetmacro{\labellength}{\labellengthin}
	}{
		\pgfmathsetmacro{\labellength}{\labellengthout}
	}
       
       	\pgfmathsetmacro{\labelAngle}{(\firstAngle + \secondAngle) / 2};
	\pgfmathsetmacro{\labelxshift}{sin(270-\labelAngle)*\labellength};
	\pgfmathsetmacro{\labelyshift}{cos(270-\labelAngle)*\labellength};
	
	\path (#1) +(-\labelxshift pt, -\labelyshift pt) node [fill=white] {\phantom{.}} node {#4} ;
}

\newcommand{\theoremlikeParameterised}[2]{\par\medskip\penalty-250{\bfseries\scshape\noindent#1 #2.}\slshape}
\newenvironment{theoremParam}[1]{\theoremlikeParameterised{\bf Theorem}{#1}}{}
\newenvironment{propositionParam}[1]{\theoremlikeParameterised{\bf Proposition}{#1}}{}

\newcommand{\linefill}{%
    \cleaders
    \hbox{$\smash{\mkern-2mu\mathord-\mkern-2mu}$}%
    \hfill
    \vphantom{\lower1pt\hbox{$\rightarrow$}}%
}

\newcommand{\xleftnoend}[1][]{\mathrel-_{\vphantom{#1}}\mkern-11mu}

\newcommand{\xleftArrow}[1][]{\leftarrow_{\vphantom{#1}}\mkern-11mu}
\newcommand{\xmid}[2][]{\stackrel{#2}{\linefill_{\vphantom{#1}}}}

\newcommand{\xrightnoend}[1][]{\mkern-11mu\mathrel-_{#1}}
\newcommand{\xrightArrow}[1][]{\mkern-11mu\rightarrow_{#1}}

\newcommand{\xrightArrowd}[1][]{\mkern-11mu\twoheadrightarrow_{#1}}
\newcommand{\xleftArrowd}[1][]{\twoheadleftarrow_{\vphantom{#1}}\mkern-11mu}

\newcommand{\xmake}[1]{\mathrel{\lower0pt\hbox{$#1$}}}

\newcommand{\stepext}[2][]{\xmake{\xleftnoend[#1]\xmid[#1]{#2}\xrightArrow[#1]}}%
\newcommand{\bstepext}[2][]{\xmake{\xleftArrow[#1]\xmid[#1]{#2}\xrightnoend[#1]}}%
\newcommand{\stepextd}[2][]{\xmake{\xleftnoend[#1]\xmid[#1]{#2\,}\xrightArrowd[#1]}}%
\newcommand{\bstepextd}[2][]{\xmake{\xleftArrowd[#1]\xmid[#1]{#2\,}\xrightnoend[#1]}}%
\newcommand{\stepexts}[3][]{\xmake{\xleftnoend[#1]\xmid[#1]{#2\,}\xrightArrow[#1]}_{\!#3}}%

\newcommand{\join}[1]{\stepextd{#1}\bstepextd{#1}}
\newcommand{\range}[1]{\text{\rm spt}(#1)}

\begin{document}

\mainmatter
\date{\today}
\title{Confluence Reduction for Probabilistic Systems \ifextended(extended version)\vspace{-0.2cm}\fi}
\author{Mark Timmer \and Mari\"elle Stoelinga \and Jaco van de Pol\thanks{This research has been partially funded by NWO under grant 612.063.817 (SYRUP) and grant Dn 63-257 (ROCKS), and by the European Union under FP7-ICT-2007-1 grant 214755 (QUASIMODO).\ifextended\vspace{-0.1cm}\fi}
}
\institute{Formal Methods and Tools, Faculty of EEMCS\\University of Twente, The Netherlands\\\email{\{timmer, marielle, vdpol\}@cs.utwente.nl}\ifextended\vspace{-0.2cm}\fi
}

\pagestyle{plain}
\maketitle

\begin{abstract}
This paper presents a novel technique for state space reduction of probabilistic specifications, based on a newly developed notion of confluence for probabilistic automata. We prove that this reduction preserves branching probabilistic bisimulation and can be applied on-the-fly. To support the technique, we introduce a method for detecting confluent transitions in the context of a probabilistic process algebra with data, facilitated by an earlier defined linear format. 
A case study demonstrates that significant reductions can be obtained.\ifextended\vspace{-0.1cm}\fi
 \end{abstract}

\thispagestyle{plain}

\section{Introduction}
Model checking of probabilistic systems is getting more and more attention, but there still is a large gap between the number of techniques supporting traditional model checking and those supporting probabilistic model checking. Especially methods aimed at reducing state spaces are greatly needed to battle the omnipresent state space explosion.

In this paper, we generalise the notion of confluence~\cite{Mil89} from
labelled transition systems (LTSs) to probabilistic automata (PAs)~\cite{Seg95b}. 
Basically, we define under which conditions unobservable transitions (often called $\tau$-transitions) do not influence a PA's
behaviour (i.e., they commute with all other transitions). Using this new notion of probabilistic confluence, we introduce a symbolic technique that reduces PAs while preserving branching probabilistic bisimulation. 
\\*[5pt]\noindent\textit{The non-probabilistic case.}  Our methodology follows the
approach for LTSs from~\cite{BP02}.  It consists of the
following steps: (i)~a system is specified as the parallel composition
of several processes with data; (ii)~the specification is linearised to a
canonical form that facilitates symbolic manipulations; (iii)~first-order logic formulas are generated to check symbolically which $\tau$-transitions are confluent; (iv)~an LTS is generated in such a way that confluent $\tau$-transitions are given priority, leading to an on-the-fly (potentially exponential) state space
reduction. Refinements by~\cite{PLM03} make it even possible to
perform confluence detection \mbox{on-the-fly} by means of boolean
equation systems.
\\*[5pt]\noindent\textit{The probabilistic case.} 
After recalling some basic concepts from probability theory and probabilistic automata, we introduce three novel notions of probabilistic confluence. Inspired by~\cite{Blom01}, these are \emph{weak probabilistic confluence}, \emph{probabilistic confluence} and \emph{strong probabilistic confluence} (in decreasing order of reduction power, but in increasing order of detection efficiency).

We prove that the stronger notions imply the weaker ones, and that $\tau$-tran\-si\-tions that are confluent according to any of these notions always connect branching probabilistically bisimilar states. Basically, this means that they can be given priority without losing any behaviour.
Based on this idea, we propose a reduction technique using weak probabilistic confluence, which merges all states that can reach each other by traversing only confluent transitions.
Additionally, we propose a reduction technique that can be applied using the two stronger notions of confluence. As opposed to the first technique it does not need to merge states; rather, it chooses a representative state that has all relevant behaviour. We prove that both reduction techniques yield a  branching probabilistically bisimilar PA. Therefore, they preserve virtually all interesting temporal properties.

As we want to analyse systems that would normally be too large, we need to detect confluence symbolically and use it to reduce on-the-fly during state space generation. That way, the unreduced PA never needs to be generated.~Since we have not found an efficient method for detecting (weak) probabilistic confluence, we only provide a detection method for strong probabilistic confluence. Here, we exploit a previously defined probabilistic process-algebraic linear format, which is capable of modelling any system consisting of parallel components with data~\cite{KPST10}. In this paper, we show how symbolic $\tau$-transitions can be proven confluent by solving formulas in first-order logic over this format. As a result, confluence can be detected symbolically, and the reduced PA can be generated on-the-fly.  We present a case study of leader election protocols, showing significant reductions.
\\*[5pt]\noindent\textit{Related work.} As mentioned before, we basically generalise the techniques presented in~\cite{BP02} to probabilistic automata. 

In the probabilistic setting, several reduction techniques similar to ours \mbox{exist}. Most of these are generalisations of the well-known concept of partial-order reduction (POR)~\cite{Pel93}.
In~\cite{BGC04} and~\cite{AN04}, the concept of POR was lifted to Markov decision processes, providing reductions that preserve quantitative LTL$\setminus$X. This was refined in~\cite{BDG06} to probabilistic CTL, a branching logic.
Recently, a revision of POR for distributed schedulers was introduced and implemented in PRISM~\cite{GDF09}. 

Our confluence reduction differs from these techniques on several accounts. First, POR is  applicable on state-based systems, whereas our confluence reduction is the first technique that can be used for action-based systems. As the transformation between action- and state-based blows up the state space~\cite{NV90}, having confluence reduction really provides new possibilities. Second, the definition of confluence is quite elegant, and (strong) confluence seems to be of a more local nature (which makes the correctness proofs easier). 
Third, the detection of POR requires language-specific heuristics, whereas confluence reduction 
acts at a more semantic level and can be implemented by a generic theorem prover. (Alternatively, decision procedures for a fixed set of data types could be devised.)

Our case study shows that the reductions obtained using probabilistic confluence are comparable to the reductions obtained by POR~\cite{Markus}.

\section{Preliminaries}

Given a set $S$, an element $s \in S$ and an equivalence relation $R \subseteq S \times S$, we write~$[s]_R$ for the \emph{equivalence class} of $s$ under $R$, i.e., $[s]_R = \{s' \in S \mid (s, s') \in R\}$. We write $S/R = \{[s]_R \mid s \in S\}$ for the set of all equivalence classes in~$S$.

\subsection{Probability theory and probabilistic automata}

\begin{definition}[Probability distributions]
A \emph{probability distribution} over a countable set $S$ is a function $\mu \colon S \rightarrow [0,1]$ such that $\sum_{s \in S} \mu(s) = 1$. Given \mbox{$S' \subseteq S$}, we write $\mu(S')$ to denote $\sum_{s' \in S'} \mu(s')$. We use $\distr(S)$ to denote the set of all probability distributions over $S$, and $\distrs(S)$ for the set of all substochastic probability distributions over $S$, i.e., where $0 \leq \sum_{s \in S} \mu(s) \leq 1$.

Given a probability distribution $\mu$ with $\mu(s_1) = p_1$, $\mu(s_2) = p_2, \dots$~($p_i \neq 0$), we write $\mu = \{s_1 \mapsto p_1, s_2 \mapsto p_2, \dots\}$ and let $\range{\mu} = \{s_1,s_2,  \dots\}$ denote its support. For the \emph{deterministic distribution} $\mu$ determined by  $\mu(t) = 1$ we write~$\dirac{t}$.

Given an equivalence relation $R$ over $S$ and two probability distributions $\mu, \mu'$ over $S$, we say that $\mu \equiv_R \mu'$ if and only if $\mu(C) = \mu'(C)$ for all $C \in S/R$.
\end{definition}

Probabilistic automata (PAs) are similar to labelled transition systems, except that transitions do not have a fixed successor state anymore. Instead, the state reached after taking a certain transition is determined by a probability distribution~\cite{Seg95b}. The transitions themselves can be chosen nondeterministically.
\begin{definition}[Probabilistic automata]
A \emph{probabilistic automaton (PA)} is a tuple $\pa = \langle S, s^0, L, \trans \rangle$, where
$S$ is a countable set of states of which $s^0 \in S$ is initial, $L$ is a countable set of actions, and  
${\trans} \subseteq S \times L \times \distr(S)$ is a countable transition relation.
We assume that every PA contains an unobservable action $\tau \in L$.
If $(s, a, \mu) \in \trans$, we write $s \stepext{a}{\mu}$, meaning that state $s$ enables action~$a$, after which the probability to go to $s' \in S$ is $\mu(s')$. If $\mu = \dirac{t}$, we write~$s \stepext{a}{t}$.

\end{definition}
\begin{definition}[Paths and traces]
Given a PA $\pa =  \langle S, s^0, L, \trans \rangle$, we define a \emph{path} of $\pa$ to be either a finite sequence 
$ 
\finpi = s_0 \curlystep{a_1,\mu_1} s_1 \curlystep{a_2,\mu_2} s_2 \curlystep{a_3,\mu_3} \ldots \curlystep{a_n,\mu_n} s_n,
$
or an infinite sequence
$
\infpi = s_0 \curlystep{a_1,\mu_1} s_1 \curlystep{a_2,\mu_2} s_2 \curlystep{a_3,\mu_3} \ldots
$.

For finite paths we require $s_i \in S$ for all $0 \leq i \leq n$, and $s_i \stepext{a_{i+1}}{\mu_{i+1}}$ as well as $\mu_{i+1}(s_{i+1}) > 0$ for all $0 \leq i < n$. For infinite paths these properties should hold for all $i \geq 0$.
A fragment $s \curlystep{a, \mu} s'$ denotes that the transition $s \stepext{a} \mu$~was chosen from state $s$, after which the successor~$s'$ was selected by chance (so $\mu(s') > 0$).
\begin{itemize}
\item If \smash{$\pi = s_0 \curlystep{a, \dirac{s_1}} s_1 \curlystep{a, \dirac{s_2}}  \dots \curlystep{a, \dirac{s_n}}s_n$} is a path of $\pa$ ($n \geq 0$), we write \mbox{$s_0 \stepextd{a}{s_n}$}. If each transition is also allowed to be faced backwards, we write \mbox{$s_0 \tsrc{a} s_n$}. If there exists a state~$t$ such that $s \stepextd{a} t$ and $s'\stepextd{a} t$, we write $s \join{a} s'$.
\item 
We use $\prefix(\pi,i)$ to denote \mbox{$s_0 \curlystep{a_1,\mu_1} \ldots \curlystep{a_i,\mu_i} s_i$}, and $\transition(\finpi,i)$ to denote the transition $(s_{i-1}, a_i,\mu_i)$. When $\pi$ is finite we define $|\pi| = n$ and $\last(\pi) = s_n$. 
 
\item We use $\finpaths_\pa$ to denote the set of all finite paths of $\pa$, and $\finpaths_\pa(s)$ for all finite paths where $s_0 = s$.

\item A path's \emph{trace} is the sequence of actions obtained by omitting all its states, distributions and $\tau$-steps; given \mbox{$\pi = s_0 \curlystep{a_1,\mu_1} s_1 \curlystep{\tau,\mu_2} s_2 \curlystep{a_3,\mu_3} \ldots \curlystep{a_n,\mu_n} s_n$}, we denote the sequence $a_1 a_3 \dots a_n$ by $\trace(\pi)$.
\end{itemize}
\end{definition}

\subsection{Schedulers}
To resolve the nondeterminism in probabilistic automata schedulers are used~\cite{Sto02phd}. Basically, a scheduler is just a function defining for each finite path which transition to take next. The decisions of  schedulers are allowed to be \emph{randomised}, i.e., instead of choosing a single transition a scheduler might resolve a nondeterministic choice by a probabilistic choice. Schedulers can be \emph{partial}, i.e., they might assign some probability to the decision of not choosing any next transition.
\begin{definition}[Schedulers]
A \emph{scheduler} for a PA $\pa = \langle S, s^0, L, \trans \rangle$ is a function
\[
\sched \colon \finpaths_\pa \rightarrow \distr( \{\bottom\} \union \trans),
\]
such that for every $\pi \in \finpaths_\pa$ the transitions $(s, a,\mu)$ that are scheduled by $\sched$ after $\finpi$ (i.e., $\sched(\finpi)(s, a,\mu) > 0$) are indeed possible after $\finpi$, i.e., $s = \last(\pi).$
The decision of not choosing any transition is represented by $\bottom$. 
\end{definition}

We now define the notions of finite and maximal paths of a PA given a scheduler.
\begin{definition}[Paths and maximal paths] Let $\pa$ be a PA and $\sched$ a scheduler for $\pa$. Then, the set of \emph{finite paths of $\pa$ under $\sched$} is given by
\[
    \finpaths^\sched_\pa = \{ \pi \in \finpaths_\pa \mid \forall 0 \leq i  < |\pi| \qdot \sched(\prefix(\pi,i))(\transition(\pi,i+1)) > 0    \}.
\]
We define $\finpaths^\sched_\pa(s) \subseteq \finpaths^\sched_\pa$ as the set of all such paths starting in $s$.
The set of \emph{maximal paths of $\pa$ under $\sched$} is given by
\[
    \maxpaths^\sched_\pa = \{ \pi \in \finpaths^\sched_\pa \mid \sched(\pi)(\bottom) > 0 \}. 
\]
Similarly, $\maxpaths^\sched_\pa(s)$ is the set of maximal paths of $\pa$ under $\sched$ starting in $s$.
\end{definition}

We now define the behaviour of a PA $\pa$ under a scheduler $\sched$. As schedulers resolve all nondeterministic choices, this behaviour is fully probabilistic. We can therefore compute the probability that, starting from a given state $s$, the path generated by $\sched$ has some finite prefix $\pi$. This probability is denoted by $\schedprob^\sched_{\pa,s}(\pi)$.
\begin{definition}[Path probabilities]
Let $\pa$ be a PA, $\sched$ a scheduler for $\pa$, and $s$ a state of $\pa$. Then, we define the function $\schedprob^\sched_{\pa,s} \colon \finpaths_\pa(s) \rightarrow [0,1]$ by
\begin{align*}
\schedprob^\sched_{\pa,s}(s) = 1; \qquad \schedprob^\sched_{\pa,s}(\pi \curlystep{a,\mu} t) = \schedprob^\sched_{\pa,s}(\pi) \cdot \sched(\pi)(\last(\pi),a, \mu) \cdot \mu(t).
\end{align*}
\end{definition}

Based on these probabilities we can compute the probability distribution $\distrprob^\sched_\pa(s)$ over the states where a PA $\pa$ under a scheduler $\sched$ terminates when starting in state $s$. Note that $\distrprob^\sched_\pa(s)$ is potentially substochastic (i.e., the probabilities do not add up to~$1$) when $\sched$ allows infinite behaviour.
\begin{definition}[Final state probabilities]
Let~$\pa$ be a PA and $\sched$ a scheduler for~$\pa$. Then, we define the function $\distrprob^\sched_\pa \colon S \rightarrow \distrs(S)$ by 
\[
\distrprob_\pa^\sched(s) = \Big\{ s' \mapsto \sum_{\substack{\pi \in \maxpaths_\pa^\sched(s)\\\last(\pi) = s'}} \schedprob_{\pa,s}^\sched(\pi) \cdot \sched(\pi)(\bottom) \mid s' \in S \Big\} \quad\quad \forall s \in S.
\]
\end{definition}

\section{Branching probabilistic bisimulation}
The notion of branching bisimulation for non-probabilistic systems was first introduced in~\cite{GW96}. Basically, it relates states that have an identical branching structure in the presence of $\tau$-actions. 
Segala defined a generalisation of branching bisimulation for PAs~\cite{SL95}, which we present here using the simplified definitions of~\cite{Sto02phd}.
First, we intuitively explain weak steps for PAs. Based on these ideas, we then formally introduce branching probabilistic bisimulation.

\subsection{Weak steps for probabilistic automata}
As $\tau$-steps cannot be observed, we want to abstract from them. Non-prob\-a\-bi\-lis\-ti\-cal\-ly, this is done via the weak step. A state $s$ can do a weak step to $s'$ under an action $a$, denoted by $s \weakstep{a} s'$, if there exists a path $s \stepext{\tau}{s_1} \stepext{\tau}{\!} \dots  \stepext{\tau}{s_n} \stepext{a}{s'}$ with $n \geq 0$ (often, also $\tau$-steps after the $a$-action are allowed, but this will not~concern us). Traditionally, $s \weakstep{a} s'$ is thus satisfied by an \emph{appropriate path}.

\begin{figure}[b]
\tikzstyle{half}=[node distance=0.8cm]
\hfill
\subfigure[A PA $\pa$.\label{fig:weak}]{
\begin{tikzpicture}[scale=0.88, transform shape, node distance=1.6cm]
	\node[nodesmall] (s_0) {$s$};
	\node[nodesmall] (s_1) [right of=s_0] {$t_2$};
	\node[nodesmall, half] (s_2) [below of=s_1] {$t_3$};
	\node[nodesmall, half] (s_21) [above of=s_1] {$t_1$};
	
	\draw[->] (s_0) -- node [onEdge] {$\tau$} (s_1);
	\draw[->] (s_0) -- node [onEdge] {$\tau$} (s_2);
	\draw[->] (s_0) -- node [onEdge] {$\vphantom{\tau}\smash{b}$} (s_21);
	
	\node[nodesmall] (s_3) [right of=s_1] {$t_4$};
	\node[nodesmall] (s_31) [above of=s_3, half] {$s_1$};
	\node[nodesmall] (s_4) [right of=s_3] {$s_2$};
	
	\node[nodesmall] (s_6) [right of=s_2] {$s_4$};
	\node[nodesmall] (s_7) [right of=s_6] {$s_3$};
	\node[nodesmall] (s_8) [below of=s_6, half] {$s_5$};

	\draw[->] (s_1) -- node [onEdge] {$\tau$} (s_3);
        \draw[->] (s_1) -- node [onEdge] {$a$} (s_31);
        \draw[->] (s_3) edge [bend left=75] node [auto] {$a$} (s_4);

	\draw[->] (s_2) -- node [auto,pos=0.4, auto] {$\frac{1}{2}$} (s_6);
	\draw[->] (s_2) -- node [below=5pt, left=3pt, auto, swap] {$\smash{\frac{1}{2}}$} (s_8);

	\draw (s_2) + (\arclength, 0cm) arc  (0:-26.56505:\arclength);
	\path (s_2) +(\rlength,-\ulength) node {$a$} ;

	\draw[->] (s_3) -- node [auto, pos=0.4] {$\frac{1}{2}$} (s_4);
	\draw[->] (s_3) -- node [auto, below=5pt, left=3pt, swap] {$\smash{\frac{1}{2}}$} (s_7);

	\draw (s_3) + (+\arclength, 0cm) arc  (0:-26.56505:\arclength);
	\path (s_3) +(\rlength,-\ulength) node {$a$} ;
\end{tikzpicture}
}
\hfill
\subfigure[Tree of $s \protect\weakstep{a} \mu$.\label{fig:weaktree}]{
\begin{tikzpicture}[scale=0.88, transform shape, node distance=1.6cm]
	\node[nodesmall] (s_0) {$s$};
	\node[nodesmall] (s_1) [right of=s_0] {$t_2$};
	\node[nodesmall, half] (s_2) [below of=s_1] {$t_3$};
	
	\draw (s_0) + (0.9\arclength,0cm) arc  (0:-26.56505:0.9\arclength);
	\path (s_0) +(\rlength,-\ulength) node {$\tau$} ;
	
	\draw[->] (s_0) -- node [auto, pos=0.3] {$\frac{2}{3}$} (s_1);
	\draw[->] (s_0) -- node [auto, swap, below=5pt, left=5pt] {$\smash{\frac{1}{3}}$} (s_2);
	
	\node[nodesmall] (s_3) [right of=s_1] {$t_4$};
	\node[nodesmall] (s_32) [above of=s_3, half] {$s_1$};
		\node[nodesmall] (s_4) [right of=s_3] {$s_2$};
	\node[nodesmall] (s_6) [right of=s_2] {$s_4$};
	\node[nodesmall] (s_7) [right of=s_6] {$s_3$};
	\node[nodesmall] (s_8) [below of=s_6, half] {$s_5$};
	
	\draw[->, white] (s_1) -- node [auto, pos=0.4, swap,black] {$\frac{1}{2}$} (s_3);
	\draw[->, white] (s_1) -- node [auto,  above=5pt, left=3pt, black] {$\smash{\frac{1}{2}}$} (s_32);
	\draw[->] (s_1) -- node [onEdge, pos=0.65] {$\tau$} (s_3);
	\draw[->] (s_1) -- node [onEdge, pos=0.65] {$a$} (s_32);
	\draw (s_1) + (0.9\arclength,0cm) arc  (0:27.43495:0.9\arclength);
	
	\draw[->] (s_2) -- node [auto, pos=0.2] {$\frac{1}{2}$} (s_6);
	\draw[->] (s_2) -- node [auto,  below=5pt, left=3pt,swap] {$\smash{\frac{1}{2}}$} (s_8);

	\draw (s_2) + (\arclength,0cm) arc  (0:-26.56505:\arclength);
	\path (s_2) +(\rlength,-\ulength) node {$a$} ;

	\draw[->] (s_3) -- node [auto, pos=0.4] {$\frac{7}{8}$} (s_4);
	\draw[->] (s_3) -- node [auto,  below=5pt, left=3pt,swap] {$\smash{\frac{1}{8}}$} (s_7);

	\draw (s_3) + (\arclength,0cm) arc  (0:-26.56505:\arclength);
	\path (s_3) +(\rlength,-\ulength) node {$a$} ;
\end{tikzpicture}
}\hfill{}
\caption{Weak steps.}
\label{fig:weaksteps}
\end{figure}
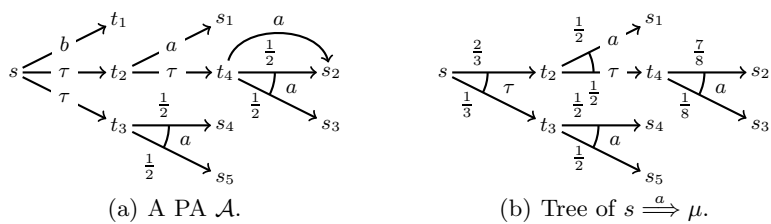

In the probabilistic setting, $s \weakstep{a} \mu$ is satisfied by an \emph{appropriate scheduler}. A scheduler $\sched$ is appropriate 
if for every  maximal path~$\pi$ that is scheduled from~$s$ with non-zero probability,
$\trace(\pi) = a$ and the $a$-transition is the last tran\-si\-tion of the path. Also, the final state distribution $\distrprob^\sched_\pa(s)$ must be equal~to~$\mu$. 
%
\begin{example}
Consider the PA shown in Figure~\ref{fig:weak}. We demonstrate that \mbox{$s \weakstep{a} \mu$}, with 
$ \mu = \{s_1 \mapsto \tfrac{8}{24}, s_2 \mapsto \tfrac{7}{24}, s_3 \mapsto \tfrac{1}{24}, s_4 \mapsto\tfrac{4}{24}, s_5 \mapsto \tfrac{4}{24} \}$. 
Take the scheduler~$\sched$:
\begin{align*}
\sched(s) &= \{ (s,\tau, \dirac{t_2}) \mapsto 2/3, (s,\tau, \dirac{t_3}) \mapsto  1/3 \}\\
\sched(t_2) &= \{(t_2, a, \dirac{s_1}) \mapsto 1/2, (t_2,\tau, \dirac{t_4}) \mapsto 1/2 \}\\
\sched(t_3) &= \{(t_3, a, \{s_4 \mapsto 1/2, s_5 \mapsto 1/2\}) \mapsto  1\}\displaybreak[0]\\
\sched(t_4) &= \{(t_4, a, \dirac{s_2}) \mapsto 3/4, (t_4, a, \{s_2 \mapsto 1/2, s_3 \mapsto 1/2 \}) \mapsto 1/4 \}\\
\sched(t_1) &= \sched(s_1) = \sched(s_2) = \sched(s_3) = \sched(s_4) = \sched(s_5) = \dirac{\bottom}
\end{align*}
Here we used $\sched(s)$ to denote the choice made for every possible path ending~in~$s$.

The scheduler is depicted in Figure~\ref{fig:weaktree}. Where it chooses probabilistically between two transitions with the same label, this is represented as a \emph{combined transition}. For instance, from $t_4$ the transition $(t_4, a, \{s_2 \mapsto  1\})$ is selected with probability $3/4$, and $(t_4, a,\{s_2 \mapsto 1/2, s_3 \mapsto 1/2\})$ with probability~$1/4$.  This corresponds to the combined transition $(t_4, a,\{s_2 \mapsto 7/8, s_3 \mapsto 1/8\})$.

Clearly, all maximal paths enabled from $s$ have trace $a$ and end directly after their $a$-transition.
The path probabilities can also be calculated. For instance, 
\begin{align*}
\schedprob^\sched_{\pa, s}(s \curlystep{\tau,\{t_2 \mapsto 1\}} t_2 \curlystep{\tau, \{t_4 \mapsto 1\}} t_4 \curlystep{a, \{s_2 \mapsto 1\}} s_2) &= \left(\tfrac{2}{3} \cdot 1\right) \cdot \left( \tfrac{1}{2} \cdot 1 \right) \cdot \left( \tfrac{3}{4} \cdot 1 \right) = \tfrac{6}{24}
\\
\schedprob^\sched_{\pa, s}(s \curlystep{\tau,\{t_2 \mapsto 1\}} t_2 \curlystep{\tau, \{t_4 \mapsto 1\}} t_4 \curlystep{a, \{s_2 \mapsto 1/2, s_3 \mapsto 1/2\}} s_2) &= \left(\tfrac{2}{3} \cdot 1\right) \cdot \left( \tfrac{1}{2} \cdot 1 \right) \cdot \left( \tfrac{1}{4} \cdot \tfrac{1}{2} \right) = \tfrac{1}{24}
\end{align*}
As no other maximal paths from $s$ go to $s_2$, 
$\distrprob_\pa^\sched(s)(s_2) = 
\tfrac{6}{24} + \tfrac{1}{24} = \tfrac{7}{24} = \mu(s_2)$.
Similarly, it can be shown that $\distrprob^\sched_\pa(s)(s_i) = \mu(s_i)$ for $i \in \{1,3,4,5\}$, so indeed it holds that $\distrprob^\sched_\pa(s) = \mu$. 
\qed
\end{example}

\subsection{Branching probabilistic bisimulation}
Before introducing branching probabilistic bisimulation, we need a restriction on weak steps. Given an equivalence relation $R$, we let $s \weakstep{a}_R \mu$ denote~that $(s,t)\in R$ for every state $t$ before the $a$-step in the tree corresponding to~$s \weakstep{a} \mu$.
\begin{definition}[Branching steps]\label{def:bbpijl}
Let $\pa = \langle S, s^0, L ,\trans \rangle$ be a PA, $s \in S$, and~$R$ an equivalence relation over $S$. Then, $s \weakstep{a}_R \mu$ if either (1) $a = \tau$ and $\mu = \dirac{s}$, or (2)~there exists a scheduler $\sched$ such that $\distrprob_\pa^\sched(s) = \mu$ and for every maximal path $s \curlystep{a_1,\mu_1} s_1 \curlystep{a_2,\mu_2} s_2 \curlystep{a_3,\mu_3} \ldots \curlystep{a_n,\mu_n} s_n \in \maxpaths^\sched_\pa(s)$ it holds that $a_n = a$, as well as $a_i = \tau$ and $(s,s_i) \in R$ for all $1 \leq i < n$.
\end{definition}

%
\begin{definition}[Branching probabilistic bisimulation]
Let $\pa\! =\! \langle S, s^0, L ,\trans \rangle$ be a PA, then an equivalence relation $R \subseteq S \times S$ is a \emph{branching probabilistic bisimulation for $\pa$} if for all $(s,t) \in R$
\[s \stepext{a}{\mu} \text{ implies } \exists \mu' \in \distr(S) \qdot t \weakstep{a}_R \mu' \en \mu \equiv_R \mu'. \]
We say that $p,q \in S$ are \emph{branching probabilistically bisimilar}, denoted \mbox{$p \bb q$}, if there exists a branching probabilistic bisimulation $R$ for $\pa$ \mbox{such that $(p,q) \in R$.}

Two PAs are branching probabilistically bisimilar if their initial states are (in the disjoint union of the two systems; see Remark 5.3.4 of~\cite{Sto02phd} for the details). 
\end{definition}
This notion has some appealing properties. First, the definition is robust in the sense that it can be adapted to using $s \weakstep{a}_R \mu$ instead of $s \stepext{a}{\mu}$ in its condition. Although this might seem to strengthen the concept, it does not. Second, the relation $\bb$ induced by the definition is an equivalence relation.
\newcommand{\propweakdef}{%
Let $\pa =  \langle S, s^0, L ,\trans \rangle$ be a PA. Then, an equivalence relation $R \subseteq S \times S$ is a branching probabilistic bisimulation for $\pa$ iff for all $(s,t) \in R$
\[s \weakstep{a}_R \mu \text{ implies } \exists \mu' \in \distr(S) \qdot t \weakstep{a}_R \mu' \en \mu \equiv_R \mu'. \]
}
\begin{proposition}\label{propweakdef}
\propweakdef
\end{proposition}
\newcommand{\eqrel}{The relation $\bb$ is an equivalence relation.}
\begin{proposition}\label{eqrel}
\eqrel
\end{proposition}

Moreover, Segala showed that branching bisimulation preserves all properties that can be expressed in the probabilistic temporal logic WPCTL (provided that no infinite path of $\tau$-actions can be scheduled with non-zero probability)~\cite{SL95}.

\section{Confluence for probabilistic automata}
The reductions we introduce are based on sets of \emph{confluent} $\tau$-transitions. Basically, such transitions do not influence a system's behaviour, i.e., a confluent step $s \stepext{\tau}{s'}$ implies that $s \bb s'$. Confluence therefore paves the way to state space reductions modulo branching probabilistic bisimulation (e.g., by giving confluent $\tau$-transitions priority). Note that not all $\tau$-transitions connect bisimilar states; even though their actions are unobservable, $\tau$-steps might disable behaviour. The aim of our analysis is to 
underapproximate which $\tau$-transitions are confluent. 

For non-probabilistic systems, several notions of confluence already exist~\cite{Blom01}. Basically, they all require that if an action~$a$ is enabled from a state that also enables a confluent $\tau$-transition, then (1) $a$ will still be enabled after taking that~$\tau$-tran\-si\-tion (possibly requiring some additional confluent $\tau$-transitions first), and (2) we can always end up in the same state traversing only confluent $\tau$-steps, no matter whether we started by the $a$- or the $\tau$-transition.

\begin{figure}[b!]
\hfill
\subfigure[Weak confluence.\label{fig:nonprobweak}]{
\begin{tikzpicture}[auto, transform shape, scale=0.88, node distance=2.1cm]
	\node[nodesmall] (s_0) {$\bullet$};
	\node[nodesmall] (s_1) [below of=s_0] {$\bullet$};
	\node[nodesmall] (s_2) [right of=s_0] {$\bullet$};
	\node[nodesmall] (s_3) [below of=s_2, node distance=0.7cm] {$\bullet$};
	\node[nodesmall] (s_4) [below of=s_3, node distance=0.7cm] {$\bullet$};
	\node[nodesmall] (s_5) [below of=s_4, node distance=0.7cm] {$\bullet$};
	
	\draw[->] (s_0) -- node [swap] {\ \ $a$} (s_1);
	\draw[-todouble] (s_0) -- node [] {$\ctau$} (s_2);
	\draw[-todouble, dashed] (s_2) -- node [] {$\ctau$} (s_3);
	\draw[->, dashed] (s_3) -- node [] {$\bar{a}$\ \ } (s_4);
	\draw[-todouble, dashed] (s_4) -- node [] {$\ctau$} (s_5);
	\draw[-todouble, dashed] (s_1) -- node [] {$\ctau$} (s_5);

\end{tikzpicture}
}
\hfill
\subfigure[Confluence.\label{fig:conf}]{
\begin{tikzpicture}[auto,  transform shape, scale=0.88,node distance=2.1cm]
	\node[nodesmall] (s_0) {$\bullet$};
	\node[nodesmall] (s_1) [below of=s_0] {$\bullet$};
	\node[nodesmall] (s_2) [right of=s_0] {$\bullet$};
	\node[nodesmall] (s_4) [below of=s_2, node distance=1.4cm] {$\bullet$};
	\node[nodesmall] (s_5) [below of=s_4, node distance=0.7cm] {$\bullet$};
	
	\draw[->] (s_0) -- node [swap] {\ \ $a$} (s_1);
	\draw[-todouble] (s_0) -- node [] {$\ctau$} (s_2);
	\draw[->, dashed] (s_2) -- node [] {$\bar{a}$\ \ } (s_4);
	\draw[-todouble, dashed] (s_4) -- node [] {$\ctau$} (s_5);
	\draw[-todouble, dashed] (s_1) -- node [] {$\ctau$} (s_5);
\end{tikzpicture}
}
\hfill
\subfigure[Strong confluence.\label{fig:strong}] {
\begin{tikzpicture}[auto,  transform shape, scale=0.88,node distance=2.1cm]
	\node[nodesmall] (s_0) {$\bullet$};
	\node[nodesmall] (s_1) [below of=s_0] {$\bullet$};
	\node[nodesmall] (s_2) [right of=s_0] {$\bullet$};
	\node[nodesmall] (s_5) [below of=s_2, node distance=2.1cm] {$\bullet$};
	
	\draw[->] (s_0) -- node [swap] {$\ \ \ a$} (s_1);
	\draw[->] (s_0) -- node [] {$\ctau$} (s_2);
	\draw[->, dashed] (s_2) -- node [] {$\bar{a}$\ \ \ } (s_5);
	\draw[->, dashed] (s_1) -- node [] {$\overline{\ctau}$} (s_5);
\end{tikzpicture}
}
\vspace{-0.14cm}\hfill {}
\caption{Three variants of confluence.}
\label{fig:confluence}
\end{figure}
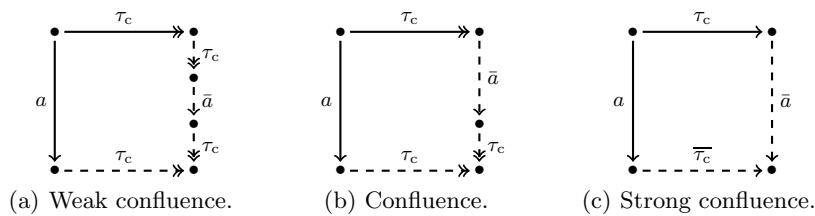

Figure~\ref{fig:confluence} depicts the three notions of confluence we will generalise~\cite{Blom01}. 
They should be interpreted as follows: for any state from which the solid transitions are enabled (universally quantified), there should be a matching for the dashed transitions (existentially quantified). A double-headed arrow denotes a path of zero of more transitions with the corresponding label, and an arrow with label~$\overline{a}$ denotes a step that is optional in case $a = \tau$ (i.e., its source and target state may then coincide). The weaker the notion, the more reduction potentially can be achieved (although detection is harder). Note that we first need to find a subset of $\tau$-transitions that we believe are confluence; then, the diagrams are~checked.

For probabilistic systems, no similar notions of confluence have been defined before. The situation is indeed more difficult, as transitions do not have 
a single target state anymore. To still enable reductions based on confluence, only $\tau$-transitions with a 
unique target state might be considered confluent. The next example shows what goes wrong without this precaution. For brevity, from now on we use \emph{bisimilar} as an abbreviation for \emph{branching probabilistically bisimilar}.
\begin{example}Consider two people each throwing a die. The PA in Figure~\ref{fig:originalspec} models this behaviour given that it is unknown who throws first. The first character of each state name indicates whether the first player has not thrown yet~(X), or threw heads (H) or tails (T), and the second character indicates the same for the second player. For lay-out purposes, some states were drawn twice.

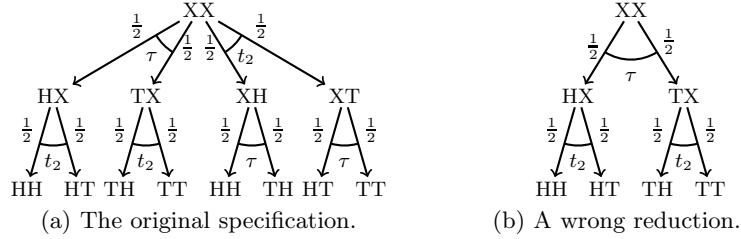
\begin{figure}[t!]
\centering
{} \hfill
\subfigure[The original specification.\label{fig:originalspec}] {
\begin{tikzpicture}[scale=0.88, transform shape, node distance=1.4cm]
\tikzstyle{half}=[node distance=0.8cm]
\tikzstyle{kwart}=[node distance=0.4cm]
	\node[nodesmall] (s0) {XX};
	\node[nodesmall] (s_1) [node distance=1.3cm, below of=s0] {};
	\node[nodesmall, half] (s2) [right of=s_1] {XH};
	\node[nodesmall, half] (s_3) [left of=s_1] {TX};
	\node[nodesmall] (s_4) [left of=s_3] {HX};
	\node[nodesmall] (s_5) [right of=s2] {XT};
	\node[nodesmall] (s_6) [below of=s2] {};
	\node[nodesmall, kwart] (s_7) [left of=s_6] {HH};
	\node[nodesmall, kwart] (s_8) [right of=s_6] {TH};
	\node[nodesmall] (s_9) [below of=s_3] {};
	\node[nodesmall, kwart] (s_10) [left of=s_9] {TH};
	\node[nodesmall, kwart] (s_11) [right of=s_9] {TT};
	\node[nodesmall] (s_12) [below of=s_5] {};
	\node[nodesmall, kwart] (s_13) [left of=s_12] {HT};
	\node[nodesmall, kwart] (s_14) [right of=s_12] {TT};
	\node[nodesmall] (s_15) [below of=s_4] {};
	\node[nodesmall, kwart] (s_16) [left of=s_15] {HH};
	\node[nodesmall, kwart] (s_17) [right of=s_15] {HT};
			
	\draw[->] (s0) -- node [above, pos=0.4] {$\frac{1}{2}$} (s_4);
	\draw[->] (s0) -- node [right, pos=0.5] {\sfrac{1}{2}} (s_3);

	\drawarc{s0}{s_4}{s_3}{$\tau$}{false};

	\draw[->] (s0) -- node [left, pos=0.5] {\sfrac{1}{2}} (s2);
	\draw[->] (s0) -- node [above, pos=0.4] {$\frac{1}{2}$} (s_5);
	\drawarc{s0}{s2}{s_5}{$t_2$}{false};
	
	\draw[->] (s_4) -- node [left, pos=0.4] {\sfrac{1}{2}} (s_16);
	\draw[->] (s_4) -- node [right, pos=0.4] {\sfrac{1}{2}} (s_17);
	\drawarc{s_4}{s_16}{s_17}{$t_2$}{false};

	\draw[->] (s_3) -- node [left, pos=0.4] {\sfrac{1}{2}} (s_10);
	\draw[->] (s_3) -- node [right, pos=0.4] {\sfrac{1}{2}} (s_11);
	\drawarc{s_3}{s_10}{s_11}{$t_2$}{false};
	
	\draw[->] (s_5) -- node [left, pos=0.4] {\sfrac{1}{2}} (s_13);
	\draw[->] (s_5) -- node [right, pos=0.4] {\sfrac{1}{2}} (s_14);
	\drawarc{s_5}{s_13}{s_14}{$\tau$}{false};

	\draw[->] (s2) -- node [left, pos=0.4] {\sfrac{1}{2}} (s_7);
	\draw[->] (s2) -- node [right, pos=0.4] {\sfrac{1}{2}} (s_8);
	\drawarc{s2}{s_7}{s_8}{$\tau$}{false};

\end{tikzpicture}
}
\hfill
\subfigure[A wrong reduction.\!\label{fig:wrongreduction}] {
\begin{tikzpicture}[scale=0.88, transform shape, node distance=1.4cm]
\tikzstyle{half}=[node distance=0.8cm]
\tikzstyle{kwart}=[node distance=0.4cm]
	\node[nodesmall] (s0) {XX};
	\node[nodesmall] (s_1) [below of=s0, node distance=1.3cm] {};
	\node[nodesmall, half] (s_3) [right of=s_1] {TX};
	\node[nodesmall, half] (s_4) [left of=s_1] {HX};
	\node[nodesmall] (s_9) [below of=s_3] {};
	\node[nodesmall, kwart] (s_10) [left of=s_9] {TH};
	\node[nodesmall, kwart] (s_11) [right of=s_9] {TT};
	\node[nodesmall] (s_15) [below of=s_4] {};
	\node[nodesmall, kwart] (s_16) [left of=s_15] {\phantom{aaa}HH\phantom{aaa}};
	\node[nodesmall, kwart] (s_17) [right of=s_15] {HT};
			
	\draw[->] (s0) -- node [left, pos=0.4] {$\frac{1}{2}$} (s_4);
	\draw[->] (s0) -- node [right, pos=0.4] {\sfrac{1}{2}} (s_3);

	\drawarc{s0}{s_4}{s_3}{$\tau$}{false};

	\draw[->] (s_4) -- node [left, pos=0.4] {\sfrac{1}{2}} (s_16);
	\draw[->] (s_4) -- node [right, pos=0.4] {\sfrac{1}{2}} (s_17);
	\drawarc{s_4}{s_16}{s_17}{$t_2$}{false};

	\draw[->] (s_3) -- node [left, pos=0.4] {\sfrac{1}{2}} (s_10);
	\draw[->] (s_3) -- node [right, pos=0.4] {\sfrac{1}{2}} (s_11);
	\drawarc{s_3}{s_10}{s_11}{$t_2$}{false};

\end{tikzpicture}
}\vspace{-0.14cm}
\hfill {}
\caption{Two people throwing dice.}
\label{fig:dices}
\end{figure}

We hid the first player's throw action, and kept the other one visible. Now, it might appear that the order in which the $a$- and the $\tau$-transition occur does not influence the behaviour. However,  the $\tau$-step does not connect bisimilar states (assuming HH, HT, TH, and TT to be distinct). After all, from state~XX it is possible to reach a state (XH) from where HH is reached with probability~$0.5$ and TH with probability~$0.5$. From HX and TX no such state is reachable anymore. Giving the $\tau$-transition priority, as depicted in Figure~\ref{fig:wrongreduction}, therefore yields a reduced system that is \emph{not} bisimilar to the original system anymore.
\qed
\end{example}

\begin{wrapfigure}[5]{r}{0.38\textwidth}
\vspace{-1.2cm}\begin{center}
\begin{tikzpicture}[scale=0.88, transform shape, node distance=1.6cm]
\tikzstyle{dubbel}=[node distance=3.2cm]
\tikzstyle{half}=[node distance=0.8cm]
\tikzstyle{kwart}=[node distance=0.4cm]
	\node[nodesmall] (s0) {$s$};
	\node[nodesmall] (s01) [right of=s0, node distance=1.4cm] {$t_0$};
	\node[nodesmall] (s1) [right of=s01, node distance=1.4cm] {$t$};
	\node[nodesmall] (s2) [below of=s0] {};
	\node[nodesmall, half] (s3) [left of=s2] {$s_1$};
	\node[nodesmall, half] (s4) [right of=s2] {$s_2$};
	
	\draw[->] (s0) -- node [left, pos=0.4] {$\mu$\ \ \sfrac{1}{2}} (s3);
	\draw[->] (s0) -- node [right, pos=0.4] {\sfrac{1}{2}} (s4);
	\drawarc{s0}{s3}{s4}{$a$}{true};
	
	\draw[->] (s0) -- node [above] {$\tau$} (s01);
	\draw[->] (s01) -- node [above] {$\tau$} (s1);
	
	\node[nodesmall] (s5) [below of=s1] {$t_2$};
	\node[nodesmall, half] (s6) [left of=s5] {$t_1$};
	\node[nodesmall, half] (s7) [right of=s5] {$t_3$};

	\draw[->] (s1) -- node [left=-0.1cm, pos=0.75] {\sfrac{1}{6}} (s5);
	\draw[->] (s1) -- node [left, pos=0.4] {\sfrac{1}{3}} (s6);
	\draw[->] (s1) -- node [right, pos=0.4] {\sfrac{1}{2} \ $\nu$} (s7);
	\drawarc{s1}{s6}{s7}{$a$}{true};
\end{tikzpicture}
\end{center}
\end{wrapfigure}
Another difficulty arises when defining probabilistic confluence. Although for LTSs~it is clear that a path $a\tau$ should reach the same state as $\tau a$, for PAs this is more involved as the $a$-step leads us to a distribution over~states. 
So, how should the model shown here be completed for the $\tau$-steps \mbox{to be~confluent?} 
%

Since we want confluent $\tau$-transitions to connect bisimilar states, we must assure that $s$, $t_0$, and $t$ are bisimilar. Therefore, $\mu$ and $\nu$ must assign equal probabilities to each \emph{class} of bisimilar states. Given the assumption that the other confluent $\tau$-transitions already connect bisimilar states, this is the case if $\mu \equiv_R \nu$ for \mbox{$R = \{(s,s') \mid s \join{\tau} s' \text{ using only confluent $\tau$-steps}\}$}.
The following definition formalises these observations. Here we use the notation $s\stepext{\ctau}{s'}$, given a set of $\tau$-transitions $c$, to denote that $s \stepext{\tau}{s'}$ and $(s,\tau,s') \in c$. 
%

We define three notions of probabilistic confluence, all requiring the target state of a confluent step
to be able to mimic the behaviour of its source state. In the weak version, mimicking
may be postponed and is based on joinability~(Definition~\ref{defweakconf}a). In the default version,
mimicking must happen immediately, but is still based on joinability~(Definition~\ref{defweakconf}b). Finally, the strong version requires immediate mimicking by directed steps~(Definition~\ref{defstrongconf}).
\begin{definition}[(Weak) probabilistic confluence]\label{defweakconf}
Let $\pa=\langle S, s^0, L, \trans  \rangle$ be a PA and $c \subseteq \{(s,a,\mu) \in \trans \mid a = \tau, \mu \text{ is deterministic}\}$ a set of $\tau$-transitions. (a) Then, $c$ is \emph{weakly probabilistically confluent} if $R = \{(s,s') \mid s \join{\ctau}{s'}\}$ is an equivalence relation, and for every path $s \stepextd{\ctau}{t}$ and all $a \in L, \mu \in \distr(S)$
\begin{align*}
    \smash{s \stepext{a}{\mu} \implies} & \smash{\exists t' \in S \qdot t \stepextd{\ctau}{t'} \en} \\
  & \smash{\left(\left(   \exists \nu \in \distr(S) \qdot t' \stepext{a}{\nu} \en \mu \equiv_{R} \nu \right) {} \of {} \left( a = \tau \en \mu \equiv_{R} \dirac{t'} \right)\right)}.
\end{align*}
(b) If for every path $s \stepextd{\ctau}{t}$ and every transition $s \stepext{a}{\mu}$ the above implication can be satisfied by taking $t' =t$, then we say that $c$ is \emph{probabilistically confluent}.
\end{definition}
For the strongest variant of confluence, moreover, we require the target states of $\mu$  to be connected by direct $\ctau$-transitions to the target states of $\nu$.
\begin{definition}[Equivalence up to $\ctau$-steps]
Let $\mu, \nu$ be two probability distributions, and let $\nu = \{t_1 \mapsto p_1, t_2 \mapsto p_2, \dots\}$. Then, $\mu$ is \emph{equivalent to $\nu$ up to $\ctau$-steps}, denoted by $\mu \ctauequiv \nu$, if there exists a partition $\range{\mu} = \biguplus_{i=1}^n S_i$ such that $n = |\range{\nu}|$ and $\forall 1 \leq i \leq n \colon \mu(S_i) = \nu(t_i) \en \forall s \in S_i \colon s \stepext{\ctau} t_i$.
\end{definition}
\begin{definition}[Strong probabilistic confluence]\label{defstrongconf}
Let $\pa=\langle S, s^0, L, \trans \rangle$ be~a PA and $c \subseteq \{(s,a,\mu) \in \trans \mid a = \tau, \mu \text{ is deterministic}\}$ a set of $\tau$-transitions,~then $c$ is \emph{strongly probabilistically confluent} if for all $s \stepext{\ctau}{t}, a \in L, \mu \in \distr(S)$
\[
     s \stepext{a}{\mu} \implies \left(\left(\exists \nu \in \distr(S) \qdot t \stepext{a}{\nu} \en \smash{\mu  \ctauequiv \nu} \right) \of \left(a = \tau \en \mu = \dirac{t} \right)\right).
\]
\end{definition}


\newcommand{\propstrongconfweak}{%
Strong probabilistic confluence implies probabilistic confluence, and probabilistic confluence implies weak probabilistic confluence.
}
\begin{proposition}\label{prop:strongimpliesweak}
\propstrongconfweak
\end{proposition}
 A transition $s \stepext{\tau} t$ is called (weakly, strongly) probabilistically confluent if there exists a (weakly, strongly) probabilistically confluent set $c$ such that $(s,\tau,t) \in c$.
\begin{example}
Observe the PAs in Figure~\ref{fig:PAstrongconfluence}. Assume that all transitions of $s$, $t_0$ and $t$ are shown, and that all $s_i, t_i,$ are potentially distinct. We marked all $\tau$-transitions as being confluent, and will verify this for some of them.
%
\begin{figure}[t!]
{}\hfill \subfigure[Weak probabilistic confluence.\label{fig:weakconf}]{
\begin{tikzpicture}[scale=0.89, transform shape, node distance=1.6cm]
\tikzstyle{dubbel}=[node distance=3.2cm]
\tikzstyle{half}=[node distance=0.8cm]
\tikzstyle{kwart}=[node distance=0.4cm]
	\node[nodesmall] (s0) {$s$};
	\node[nodesmall] (s01) [right of=s0] {$t_0$};
	\node[nodesmall] (s1) [right of=s01] {$t$};
		\node[nodesmall] (s2) [below of=s0] {};
	\node[nodesmall, half] (s3) [left of=s2] {$s_1$};
	\node[nodesmall, half] (s4) [right of=s2] {$s_2$};
	
	\draw[->] (s0) -- node [left, pos=0.4] {\smash{\raisebox{0.2cm}{$\mu$}} \sfrac{1}{2}} (s3);
	\draw[->] (s0) -- node [right, pos=0.4] {\sfrac{1}{2}} (s4);
	\drawarc{s0}{s3}{s4}{$a$}{true};
	
	\draw[->] (s0) -- node [above] {$\ctau$} (s01);
	\draw[->] (s01) -- node [above] {$\ctau$} (s1);
	
	\node[nodesmall] (s5) [below of=s1] {$t_2$};
	\node[nodesmall, half] (s6) [left of=s5] {$t_1$};
	\node[nodesmall, half] (s7) [right of=s5] {$t_3$};

	\draw[->] (s1) -- node [left=-0.1cm, pos=0.75] {\sfrac{1}{6}} (s5);
	\draw[->] (s1) -- node [left, pos=0.4] {\sfrac{1}{3}} (s6);
	\draw[->] (s1) -- node [right, pos=0.4] {\sfrac{1}{2} \smash{\raisebox{0.2cm}{$\nu$}}} (s7);
	\drawarc{s1}{s6}{s7}{$a$}{true};

	\draw[->] (s3) [bend right=30] edge node [below, pos=0.85] {$\ctau$} (s7);
	\draw[->] (s4) edge node [above, pos=0.5] {$\ctau$} (s6);
	\draw[->] (s5) [bend left=30] edge node [onEdge] {$\ctau$} (s4);
	\draw[->] (s6) edge node [above, pos=0.5] {$\ctau$} (s5);
\end{tikzpicture}
}
\hfill
\subfigure[Strong probabilistic confluence.\label{fig:strongconf}]{
\begin{tikzpicture}[scale=0.89, transform shape, node distance=1.6cm]
\tikzstyle{dubbel}=[node distance=3.2cm]
\tikzstyle{half}=[node distance=0.8cm]
\tikzstyle{kwart}=[node distance=0.4cm]
	\node[nodesmall] (s0) {$s$};
	\node[nodesmall, dubbel] (s1) [right of=s0] {$t$};
		\node[nodesmall] (s2) [below of=s0] {$s_2$};
	\node[nodesmall, half] (s3) [left of=s2] {$s_1$};
	\node[nodesmall, half] (s4) [right of=s2] {$s_3$};
	
	\draw[->] (s0) -- node [left, pos=0.4] {\smash{\raisebox{0.2cm}{$\mu$}} \sfrac{1}{3}} (s3);
	\draw[->] (s0) -- node [left=-0.1cm, pos=0.75] {\sfrac{1}{3}} (s2);
	\draw[->] (s0) -- node [right, pos=0.4] {\sfrac{1}{3}} (s4);
	\drawarc{s0}{s3}{s4}{$a$}{true};
	
	\draw[->] (s0) -- node [above] {$\ctau$} (s1);
	
	\node[nodesmall] (s5) [below of=s1] {};
	\node[nodesmall, half] (s6) [left of=s5] {$t_2$};
	\node[nodesmall, half] (s7) [right of=s5] {$t_1$};

	\draw[->] (s1) -- node [left, pos=0.4] {\sfrac{2}{3}} (s6);
	\draw[->] (s1) -- node [right, pos=0.4] {\sfrac{1}{3} \smash{\raisebox{0.2cm}{$\nu$}}} (s7);
	\drawarc{s1}{s6}{s7}{$a$}{true};
	
	\draw[->] (s3) [bend right=30] edge node [below, pos=0.85] {$\ctau$} (s7);
	\draw[->] (s2) [bend right=30] edge node [below, pos=0.85] {$\ctau$} (s6);
	\draw[->] (s4) edge node [above, pos=0.5] {$\ctau$} (s6);
\end{tikzpicture}
}\hfill{}
\caption{Weak versus strong confluence.}
\label{fig:PAstrongconfluence}
\end{figure}
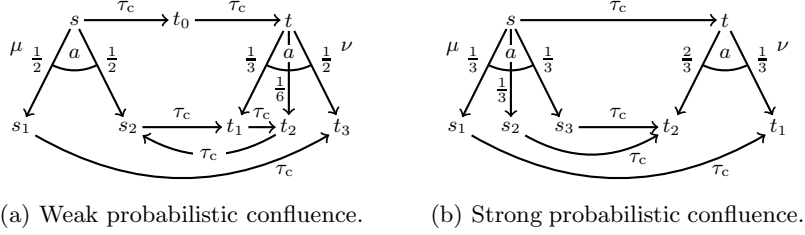

In Figure~\ref{fig:weakconf}, both the upper $\ctau$-steps are weakly probabilistically confluent, most interestingly $s \stepext{\ctau}{t_0}$. To verify this, first note that $t_0 \stepext{\ctau}{t}$ is (as $t_0$ has no other outgoing transitions), from where the $a$-transition of $s$ can be mimicked. To see that indeed $\mu \equiv_R \nu$ (using $R$ from Definition~\ref{defweakconf}), observe that $R$ yields two equivalence classes: $C_1 = \{s_2, t_1, t_2\}$ and $C_2 =\{s_1, t_3\}$. As required, $\mu(C_1) = \frac{1}{2} = \nu(C_1)$ and $\mu(C_2) = \frac{1}{2} = \nu(C_2)$.
Clearly $s \stepext{\ctau}{t_0}$ is not probabilistically confluent, as $t_0$ cannot immediately mimic the $a$-transition of~$s$. 

In Figure~\ref{fig:strongconf} the upper $\ctau$-transition is strongly probabilistically confluent (and therefore also (weakly) probabilistically confluent). For this, $t$ must be able to directly mimic the $a$-transition from $s$. Indeed, it can do so by the transition $t \stepext{a}{\nu}$. Moreover, $\smash{\mu \ctauequiv \nu}$ also holds, which is easily seen by taking the partition $S_1 = \{s_1\}, S_2 = \{s_2, s_3\}$.
\qed
\end{example}

The following theorem shows that weakly probabilistically confluent $\tau$-tran\-si\-tions indeed connect bisimilar states. With Proposition~\ref{prop:strongimpliesweak} in mind, this also holds for (strong) probabilistic confluence. Additionally, we show that confluent sets can be joined (so there is a unique maximal confluent set of $\tau$-transitions).
\newcommand{\theoremconfbranching}{%
Let $\pa = \langle S, s^0, L, \trans \rangle$ be a PA, $s,s' \in S$ two of its states, and $c$ a weakly probabilistically confluent subset of its $\tau$-transitions. Then,
\[
   s \tsrc{\ctau} s' \text{ implies } s \bb s'.
\]
}
\begin{theorem}\label{theoremconfbranching}
\theoremconfbranching
\end{theorem}
\newcommand{\propunion}{
Let $c,c'$ be (weakly, strongly) probabilistically confluent sets of $\tau$-transitions. Then, $c \union c'$ is also (weakly, strongly) probabilistically confluent.
}
\begin{proposition}\label{propunion}
\propunion
\end{proposition}

\section{State space reduction using probabilistic confluence}\label{sec:confred}
As confluent $\tau$-transitions lead from a state $s$ to a state $s'$ such that $s'$ is equivalent to $s$ (with respect to branching probabilistic bisimulation), all states that can reach each other via such transitions can be merged. That is, we can take the original PA modulo the equivalence relation~$\tsrc{\ctau}$ and obtain a reduced and bisimilar system. The next definition and theorem formally state this.
\begin{definition}[$\pa/R$]\label{defeqrel}
Let $\pa = \langle S, s^0, L, \trans \rangle$ be a PA and $R$ an equivalence relation over~$S$, then we write $\pa / R$ to denote the PA $\pa$ modulo $R$. That is,
\[\pa / R = \langle S / R, [s^0]_R, L, \trans_{\!R} \rangle, \]
with ${\trans_{\!R}} \subseteq S/R \times L \times \distr(S/R)$ such that $[s]_R \stepexts{a}{R}{\mu}$ if and only there exists a state $s' \in [s]_R$ such that $s' \stepext{a}{\mu'}$ and $\forall {[t]_R \in S/R} \qdot \mu([t]_R) = \sum_{t' \in [t]_R} \mu'(t')$.
\end{definition}
\newcommand{\theoremeqclasses}{
Let $\pa$ be a PA and $c$ a weakly probabilistically confluent subset of its $\tau$-transitions, then 
   $\left(\pa / {\tsrc{\ctau}}\right) \bb \pa$.
}
\begin{theorem}\label{theoremeqclasses}
\theoremeqclasses
\end{theorem}
The downside of this method is that, in general, it is hard to compute the equivalence classes according to $\tsrc{\ctau}$. Therefore, a slightly adapted reduction technique was proposed in~\cite{Blom01}, and later used in~\cite{BP02}. There, for each equivalence class a single representative state $s$ was chosen in such a way that all transitions leaving the equivalence class are directly enabled from $s$.
This method relies on (strong) probabilistic confluence, and does not work for the weak variant.

To find a valid representative, we first look at the directed (unlabeled) graph $G = (S,\stepext{\ctau}{\!\!})$. It contains all states of the original system, and denotes precisely which states can reach each other by taking only $\ctau$-transitions. Because of the restrictions on $\ctau$-transitions, the subgraph of $G$ corresponding to each equivalence class $[s]_{\tsrc{\ctau}}$has exactly one terminal strongly connected component (TSCC), from which the representative state for that equivalence class should be chosen. Intuitively, this follows from the fact that $\ctau$-transitions always lead to a state with at least the same observable transitions as the previous state, and maybe more. (This is not the case for weak probabilistic confluence, therefore the reduction using representatives does not work for that variant of confluence.) The next definition formalises these observations.
\begin{definition}[Representation maps]
Let $\pa$ be a PA and $c$ a subset of its $\tau$-transitions. Then, a function $\phi_c \colon S \rightarrow S$ is a \emph{representation map} for $\pa$ under~$c$~if
\begin{itemize}
\item $\forall s, s' \in S \qdot s \stepext{\ctau}{s'} \implies \phi_c(s) = \phi_c(s')$;
\item $\forall s \in S \qdot s \stepextd{\ctau} \phi_c(s)$.
\end{itemize}
\end{definition}
The first condition ensures that equivalent states are mapped to the same representative, and the second makes sure that every representative is in a TSCC. 
If $c$ is a probabilistically confluent set of $\tau$-transitions, the second condition and Theorem~\ref{theoremconfbranching} immediately imply that $s \bb \phi_c(s)$ for every state $s$.

The next proposition states that for finite-state PAs and probabilistically confluent sets $c$, there always exists a representation map. As $\ctau$-transitions are required to always have a deterministic distribution, probabilities are not involved and the proof is identical to the proof for the non-probabilistic case~\cite{Blom01}.
\begin{proposition}
Let $\pa$ be a PA and $c$ a probabilistically confluent subset of its $\tau$-transitions. Moreover, let $S_\pa$ be finite. Then, there exists a function $\phi_c \colon S \rightarrow S$ such that $\phi_c$ is a representation map for $\pa$ under~$c$.
\end{proposition}

We can now define a PA modulo a representation map $\phi_c$.  The set of states of such a PA consists of all representatives. When originally $s \stepext{a}{\mu}$ for some state $s$, in the reduced system $\phi_c(s) \stepext{a}{\mu'}$ where $\mu'$ assigns a probability to each representative equal to the probability of reaching any state that maps to this representative in the original system. The system will not have any $\ctau$-transitions.
\begin{definition}[$\pa/\phi_c$]\label{defmodulorep}
Let $\pa = \langle S, s^0, L, \trans \rangle$ be a PA and $c$ a set of $\tau$-tran\-si\-tions. Moreover, let $\phi_c$ be a representation map for $\pa$ under $c$. Then, we write $\pa / \phi_c$ to denote the PA $\pa$ modulo $\phi_c$. That is,
\[\pa / \phi_c = \langle \phi_c(S), \phi_c(s^0), L, \trans_{\phi_c} \rangle, \]
where $\phi_c(S) = \{\phi_c(s) \mid s \in S\}$, and ${\trans_{\phi_c}} \subseteq \phi_c(S) \times L \times \distr(\phi_c(S))$ such that $s \stepexts{a}{\phi_c}{\mu}$ if and only if $a \neq \ctau$ and there exists a transition $t \stepext{a}{\mu'}$ in $\pa$ such that $\phi_c(t) = s$ and
$\forall s' \in \phi_c(S) \qdot \mu(s') = \mu'(\{s'' \in S \mid \phi_c(s'') = s'\})$.
\end{definition}

From the construction of the representation map it follows that $\pa / \phi_c \bb \pa$ if $c$ is (strongly) probabilistically confluent.
\newcommand{\theoremrep}{
Let $\pa$ be a PA and $c$ a probabilistically confluent set of $\tau$-tran\-si\-tions. Also, let $\phi_c$ be a representation map for $\pa$ under~$c$. Then,
$\left(\pa / \phi_c\right) \bb \pa$.
}
\begin{theorem}\label{theoremrep}
\theoremrep
\end{theorem}
Using this result, state space generation of PAs can be optimised in exactly the same way as has been done for the non-probabilistic setting~\cite{BP02}. Basically, every state visited during the generation is replaced on-the-fly by its representative. In the absence of $\tau$-loops this is easy; just repeatedly follow confluent $\tau$-transitions until none are enabled anymore. When $\tau$-loops are present, a variant of Tarjan's algorithm for finding SCCs can be applied (see \cite{Blom01} for the details).

\section{Symbolic detection of probabilistic confluence}\label{secdetect}

Before any reductions can be obtained in practice, probabilistically confluent $\tau$-transitions need to be detected. As our goal is to prevent the generation of large state spaces, this has to be done symbolically.

We propose to do so in the framework of prCRL and LPPEs~\cite{KPST10}, where systems are modelled by a process algebra and every specification is \emph{linearised} to an intermediate format: the LPPE (linear probabilistic process equation). Basically, an LPPE is a process $X$ with a vector of global variables~$\vec{g}$ of type~$\vec{G}$ and a set of \emph{summands}. A summand is a symbolic transition that is chosen nondeterministically, provided that its guard is enabled (similar to a guarded command). Each summand $i$ is of the form
\begin{align*}
 \smash{    \sum_{\vec{d_i}:\vec{D_i}}} c_i(\vec{g}, \vec{d_i}) \Rightarrow  a_i(\vec{g}, \vec{d_i})\!\psum_{\vec{e_i}:\vec{E_i}} f_i(\vec{g}, \vec{d_i}, \vec{e_i}) \colon X(\vec{n_i}(\vec{g}, \vec{d_i}, \vec{e_i})).
\end{align*}
Here, $\vec{d_i}$ is a (possibly empty) vector of local variables of type $\vec{D_i}$, which is chosen nondeterministically such that the condition $c_i$ holds. Then, the action $a_i(\vec{g}, \vec{d_i})$ is taken and a vector $\vec{e_i}$ of type $\vec{E_i}$ is chosen probabilistically (each $\vec{e_i}$ with probability $f_i(\vec{g}, \vec{d_i}, \vec{e_i})$). Then, the next state is set to $\vec{n_i}(\vec{g}, \vec{d_i}, \vec{e_i})$.

The semantics of an LPPE is given as a PA, whose states are precisely all vectors $\vec{g} \in \vec{G}$. For all~\mbox{$\vec{g} \in \vec{G}$}, there is a transition $\vec{g} \stepext{a} \mu$ if and only if for at least one summand~$i$ there is a choice of local variables $\vec{d_i} \in \vec{D_i}$ such that
\[
    c_i(\vec{g}, \vec{d_i}) \en a_i(\vec{g}, \vec{d_i}) = a \en \forall \vec{e_i} \in \vec{E_i} \qdot \mu(\vec{n_i}(\vec{g}, \vec{d_i}, \vec{e_i})) =\!\!\!\!\!\!\!\!\!\!\!\!\!\!\!\!\!\sum_{\substack{\smash{\vec{e'_i}} \in \vec{E_i}\\\vec{n_i}(\vec{g}, \vec{d_i}, \vec{e_i}) = \vec{n_i}(\vec{g}, \vec{d_i}, \vec{e_i'})}}\!\!\!\!\!\!\!\!\!\!\!\!\!\!\!\!\!f_i(\vec{g}, \vec{d_i}, \vec{e_i'}).
\]

\begin{example} As an example of an LPPE, observe the following specification:
\begin{align*} 
 X(\textit{pc} : \{1,2\}) =& \sum_{n : \{1,2,3\}} \textit{pc}=1 \Rightarrow \text{output}(n) \psum_{\mathclap{i : \{1,2\}}} \tfrac{i}{3} \colon X(i) & (1)\\ 
+& \hphantom{\sum_{n : \{1,2,3\}}} \ \textit{pc}=2  \Rightarrow \text{beep} \psum_{\mathclap{j : \{1\}}} 1 \colon X(j) &(2)
 \end{align*}
The system has one global variable \textit{pc} (which can be either $1$ or $2$), and consists~of two summands. When $\textit{pc} = 1$, the first summand is enabled and the system~nondeterministically chooses $n$ to be~$1$, $2$ or $3$, and outputs the chosen number. Then, the next state is chosen prob\-a\-bi\-lis\-ti\-cal\-ly; with probability $\tfrac{1}{3}$~it~will be~$X(1)$, and with probability $\tfrac{2}{3}$ it will be~$X(2)$. When $\textit{pc} = 2$, the second~summand is enabled, making the system beep and deterministically returning to~$X(1)$.

In general, the conditions and actions may depend on both the global variables (in this case \textit{pc}) and the local variables (in this case $n$ for the first summand), and the probabilities and expressions for determining the next state may additionally depend on the probabilistic variables (in this case $i$ and $j$).  \qed
\end{example}

Instead of designating \emph{individual} $\tau$-transitions to be probabilistically confluent, we   designate \emph{summands} to be so in case we are sure that \emph{all} transitions they might generate are probabilistically confluent.
For a summand~$i$ to be confluent, clearly $a_i(\vec{g}, \vec{d_i}) = \tau$ should hold for all possible values of $\vec{g}$ and~$\vec{d_i}$. Also, the next state of each of the transitions it generates should be unique: for every possible valuation of $\vec{g}$ and $\vec{d_i}$, there should be a single $\vec{e_i}$ such that $f_i(\vec{g}, \vec{d_i}, \vec{e_i}) = 1$.

Moreover, a confluence property should hold. For efficiency, we detect a strong variant of strong probabilistic confluence. Basically, a confluent $\tau$-summand~$i$ has to commute properly with every summand $j$ (including~itself). More precisely, when both are enabled, executing one should not disable the other and the order of their execution should not influence the observable behaviour or the final state. Additionally, $i$ commutes with itself if it generates only one transition. Formally:
 \begin{align}\label{confluenceformule}
\begin{array}{ll}
\ifextended\else\vspace{-0.1cm}\fi
 \smash{\big(}c_i(\vec{g},\vec{d_i}) \en c_j(\vec{g},\vec{d_j})\smash{\big)} \rightarrow &\smash{\big(} i = j \en \vec{n_i}(\vec{g}, \vec{d_i}) = \vec{n_j}(\vec{g}, \vec{d_j})\smash{\big)} \of \\[5pt]
& \!\!\!\!\!\left(\begin{array}{ll}
  & c_j(\vec{n_i}(\vec{g},\vec{d_i}), \vec{d_j}) \en c_i(\vec{n_j}(\vec{g},\vec{d_j},\vec{e_j}), \vec{d_i})\\
\en & a_j(\vec{g},\vec{d_j}) = a_j(\vec{n_i}(\vec{g},\vec{d_i}),\vec{d_j})\\
\en & f_j(\vec{g}, \vec{d_j}, \vec{e_j}) = f_j(\vec{n_i}(\vec{g}, \vec{d_i}), \vec{d_j}, \vec{e_j})\\
\en & \vec{n_j}(\vec{n_i}(\vec{g},\vec{d_i}),\vec{d_j},\vec{e_j}) =\vec{n_i}(\vec{n_j}(\vec{g},\vec{d_j},\vec{e_j}),\vec{d_i})
\end{array}\right) \ifextended\else\vspace{-0.1cm}\fi
 \end{array}
\end{align} 
where $\vec{g}, \vec{d_i}, \vec{d_j}$ and $\vec{e_j}$ universally quantify over $\vec{G}$, $\vec{D_i}$, $\vec{D_j}$, and $\vec{E_j}$, respectively. We used $\vec{n_i}(\vec{g},\vec{d_i})$ to denote the unique target state of summand $i$ given global state $\vec{g}$ and local state $\vec{d_i}$ (so $\vec{e_i}$ does not need to appear).

As these formulas are quantifier-free and in practice often either trivially false or true, they can easily be solved using an SMT solver for the data types involved. Note that~$n^2$ formulas need to be solved ($n$ being the number of summands); the complexity of this depends on the data types. In our experiments, all formulas could be checked with fairly simple heuristics (such as validating them vacuously by finding contradictory conditions, or by detecting that two summands never use or change the same~variable).
\newcommand{\theoremsymbolic}{%
Let $X$ be an LPPE and $\pa$ its PA. Then, if for a summand~$i$ we have $\forall \vec{g} \in \vec{G}, \vec{d_i} \in \vec{D_i} \qdot a_i(\vec{g}, \vec{d_i}) = \tau \en \exists \vec{e_i} \in \vec{E_i} \qdot f_i(\vec{g}, \vec{d_i}, \vec{e_i}) = 1$ and~formula~(\ref{confluenceformule}) holds, the set of transitions  generated by $i$ is \mbox{probabilistically confluent.}
}
\begin{theorem}\label{theoremsymbolic}
\theoremsymbolic
\end{theorem}
\ifextended\else\vspace{-0.2cm}\fi

\section{Case study}\label{seccasestudy}\ifextended\else\vspace{-0.1cm}\fi
To illustrate the power of probabilistic confluence reduction, we applied it on the leader election protocol introduced in~\cite{KPST10}. This protocol, between two nodes, decides on a leader by having both parties throw a die and compare the results. In case of a tie the nodes throw again, otherwise the one that threw highest will be the leader. We hid all actions needed for rolling the dice and communication, keeping only the declarations of leader and follower.
The complete model in LPPE format, consisting of twelve summands, can be found in Appendix~\ref{appcase}. 

In~\cite{KPST10} we showed the effect of dead-variable reduction on this system. Now, we apply probabilistic confluence reduction both to the LPPE that was already reduced in this way (\texttt{leaderReduced}) and to the original one (\texttt{leader}). To do this automatically, we implemented a prototype tool in Haskell for confluence detection and reduction using heuristics\footnote{The implementation, case studies and a test script can be downloaded from 
\mbox{\url{http://fmt.cs.utwente.nl/~timmer/papers/tacas2011.html}}.
}, relying on Theorem~\ref{theoremsymbolic}.

We used confluence information when generating the state space, applying Theorem~\ref{theoremrep}. As the specification does not contain loops of confluent $\tau$-summands, we could from each state repeatedly execute confluent $\tau$-summands until reaching a state 
that does not enable any confluent $\tau$-summand anymore,
 adding only this state to the state space (so no detection of TSCCs was needed).

 The results, obtained on a 2.4 GHz, 2 GB Intel Core 2 Duo MacBook, are shown in Table~\ref{tabel}; we list the size of the original and reduced state space, as well as the number of states and transitions that were visited during its generation using confluence.
%
Probabilistic confluence reduction clearly has quite an effect on the size of the state space, as~well as the number of visited states and therefore the running time. Notice that it nicely works \mbox{hand-in-hand} with dead-variable reduction. Applying both, we reduced by almost an order of~magnitude.

We also modeled another leader election protocol that uses asynchronous channels and allows for more parties (Algorithm $\mathcal{B}$ from~\cite{FP05}). We looked at either~$2$, $3$ or $4$ parties, who throw either a normal die or one with more or less sides ($5$, $12$, $18$, $19$, $36$). Confluence reduction reduces the state space by~about~$65\%$, and the number of visited states (and therefore the running time) by about~$50\%$. With probabilistic POR, comparable results were obtained for similar protocols~\cite{Markus}. As was to be expected, 
detecting confluence mostly pays off for the larger state~spaces. Still, confluence detection only took a fraction of a second for each system; practically all the effort is in the state space~generation.
From about $180000$ states swapping occurs, explaining the excessive growth~in~running time.
Confluence reduction clearly allows us to do more before reaching~this~limit.


\begin{table}[t!]
\centering
\caption{Applying confluence reduction to two leader election protocols.}
\smaller\label{tabel}
\begin{tabular}{l@{\ \ }|c@{\ \ }r@{\quad}r@{\ \ }|r@{\ \ }r@{\quad}r@{\ \ }|r@{\ \ }r@{\quad}r@{\ \ }|r@{\ \ }r@{\quad}r@{\ \ }}
 & & \multicolumn{2}{c|}{Original\ \ } && \multicolumn{2}{c|}{Reduced \ \ } &&\multicolumn{2}{c|}{Visited \ \ } &&\multicolumn{2}{c}{Running time \ \ } \\
Specification & & States  & Trans. && States & Trans. && States & Trans. && Before & After\\
\hline
 \texttt{leader} && $3763$ &$6158$ && $1399$ & $1922$ && $1471$ & $4022$ && $0.49$ sec & $0.35$ sec\\
  \texttt{leaderReduced} && $1693$ &$2438$ && \phantom{0}$589$ & \phantom{0}$722$ && $661$ & $1382$ && $0.22$ sec & $0.13$ sec\\
  \hline
\texttt{leader-2-6} && $535$ & $710$ && $199$ & $212$ && $271$ & $512$ && $0.15$ sec & $0.18$ sec\\
\texttt{leader-2-36} &&$18325$ & $23690$ && $6589$ & $ 6662$ && $9181$ & $17102$ && $13.23$ sec & $7.38$ sec\\
\texttt{leader-3-12} && $161803$ & $268515$ && $56839$ & $68919$ && $84059$ & $158403$ && $70.31$ sec & $39.50$ sec\\
\texttt{leader-3-18} && $533170$ & $880023$ && $188287$ & $226011$ && $276692$ & $518991$ && $471.42$ sec & $343.92$ sec\\
\texttt{leader-3-19} && \multicolumn{2}{l|}{\text{\,out of memory}} && $220996$ & $264996$ && $324544$ & $608433$ && $-$ & $379.19$ sec\\
\texttt{leader-4-5} && $443840$ & $939264$ && $128553$ & $200312$ && $206569$ & $418632$ && $467.69$ sec & $93.36$ sec
\end{tabular}
\end{table}

\section{Conclusions}
This paper introduced three new notions of confluence for probabilistic automata. We first established several facts about these notions, most importantly that they identify branching probabilistically bisimilar states. Then, we showed how probabilistic confluence can be used for state space reduction. As we used representatives in terminal strongly connected components, these reductions can even be applied to systems containing $\tau$-loops. We discussed how confluence can be detected in the context of a probabilistic process algebra with data by proving formulas in first-order logic. This way, we enabled on-the-fly reductions when generating the state space corresponding to a process-algebraic specification. A case study illustrated the power of our methods.

\bibliographystyle{plain}
\bibliography{references}

\ifextended
\appendix
\section{Proofs}
\subsection{Proof of Proposition~\ref{propweakdef}}
\begin{propositionParam}{\ref{propweakdef}}
\propweakdef
\end{propositionParam}
\begin{proof}
If every weak step $s \weakstep{a}_R \mu$ can be mimicked, then also every step $s \stepext{a}{\mu}$ can be mimicked. After all, from $s \stepext{a}{\mu}$ it follows that $s \weakstep{a}_R \mu$ for any $R$ (by taking a scheduler that chooses the transition $(s, a, \mu)$ with probability $1$ from $s$, and chooses $\bottom$ with probability $1$ for all other histories). Therefore, the definition given in this proposition is at least as restrictive as the original definition.

Conversely, we show that when every step $s \stepext{a}{\mu}$ can be mimicked, then also every weak step $s \weakstep{a}_R \mu$ can be mimicked. When $a =\tau$ and $\mu = \dirac{s}$ this weak step can be mimicked trivially by $t \weakstep{\tau}_R \dirac{t}$. Therefore, from now on we  assume that  there exists a scheduler $\sched$ such that $\distrprob_\pa^\sched(s) = \mu$, and for every maximal path $s \curlystep{a_1,\mu_1} s_1 \curlystep{a_2,\mu_2} s_2 \curlystep{a_3,\mu_3} \ldots \curlystep{a_n,\mu_n} s_n \in \maxpaths_\pa(s)$
\begin{itemize}
\item $a_i = \tau$ and $(s,s_i) \in R$ for all $1 \leq i < n$;
\item $a_n = a$.
\end{itemize}
As every single transition can be mimicked by $t$, we can define a scheduler $\sched'$ that mimics every choice of $\sched$. So, when $\sched$ chooses the transition $(s,a_1, \mu_1)$ with probability $p$, we let $\sched'$ schedule the transitions necessary for $t \weakstep{a_1}_R \mu'_1$ (with $\mu_1 \equiv_R \mu_1'$) with probability $p$. That is, when for instance $t \stepext{a_1}{t_1}$ and $t \stepext{a_1}{t_2}$ should both be assigned probability $0.5$ to yield $t \weakstep{a_1}_R \mu'_1$, we let $\sched'$ choose them with probability $0.5p$. This way, with probability $p$ the tree starting from~$t$ reaches a distribution over states that is $R$-equivalent to $\mu$. As we can then again mimic the transitions of $\sched$ from there, and this can continue until the end of each maximal path of $\sched$, we obtain a scheduler $\sched'$ for which $\distrprob_\pa^{\sched'}(t) = \mu'$ with $\mu \equiv_R \mu'$. Moreover, all the states visited before the $a$-actions in the tree starting from $t$ also remain in the same $R$ equivalence class because of the restrictions of the $\weakstep{}_R$ relation and the fact that the mimicked steps should yield an $R$-equivalent distribution. Therefore, indeed $t \weakstep{a}_R \mu' \en \mu \equiv_R \mu'$. \qed
\end{proof}

\subsection{Proof of Proposition~\ref{eqrel}}
Before proving Proposition~\ref{eqrel}, we first provide a definition and two lemmas.
\begin{definition}[Relation composition]Given two relations $R_1$ and $R_2$ over a set $S$, we use $R_2 \comp R_1$ to denote their \emph{composition}: $R_2 \comp R_1 = \{ (x,z) \in S \times S \mid \exists y \in S \qdot (x,y) \in R_1, (y,z) \in R_2  \}$.
\end{definition}

\begin{lemma}\label{lemmasuperset}
Let $\pa = \langle S, s^0, L ,\trans \rangle$ be a PA, $s \in S$, and $R$ an equivalence relation over $S$. Let $R' \subseteq S \times S$ such that $R' \supseteq R$. Then
\[
s \weakstep{a}_R \mu \text{ implies } s \weakstep{a}_{R'} \mu
\]
\end{lemma}
\begin{proof}
Let $s \weakstep{a}_{R}$. If $a =\tau$ and $\mu = \dirac{s}$ then by definition $s \weakstep{a}_{R'} \mu$ for any $R'$, so from now on we assume the other case: there exists a scheduler $\sched$ such that $\distrprob_\pa^\sched(s) = \mu$, and for every path $s \curlystep{a_1,\mu_1} s_1 \curlystep{a_2,\mu_2} s_2 \curlystep{a_3,\mu_3} \ldots \curlystep{a_n,\mu_n} s_n \in \maxpaths_\pa(s)$
\begin{itemize}
\item $a_i = \tau$ and $(s,s_i) \in R$ for all $1 \leq i < n$;
\item $a_n = a$.
\end{itemize}
Now it is easy to see that the same scheduler proofs the validity of \mbox{$s \weakstep{a}_{R'} \mu$}. After all, the only thing that has to be checked when changing $R$ is that \mbox{$(s,s_i) \in R'$} still holds for all $1 \leq i < n$. However, as $(s,s_i) \in R$ is assumed and $R' \supseteq R'$, this is immediate.
\qed
\end{proof}

\begin{lemma}\label{lemmasuper2}
Let $\pa =  \langle S, s^0, L ,\trans \rangle$ be a PA. Let $R \subseteq S \times S$ be an equivalence relation such that for all $(s,t) \in R$ it holds that
\[s \weakstep{a}_R \mu \text{ implies } \exists \mu' \in \distr(S) \qdot t \weakstep{a}_R \mu' \en \mu \equiv_R \mu'. \]
Then, for every equivalence relation $R' \subseteq S \times S$ such that $R' \supseteq R$, it holds that
\[
   s \weakstep{a}_{R'} \mu \text{ implies } \exists \mu' \in \distr(S) \qdot t \weakstep{a}_{R'} \mu' \en \mu \equiv_{R'} \mu'. 
\]
\end{lemma}
\begin{proof}
Let $\pa=  \langle S, s^0, L ,\trans \rangle$ be a PA, and let $R \subseteq S \times S$ be an equivalence relation such that for all $(s,t) \in R$ it holds that
\[s \weakstep{a}_R \mu \text{ implies } \exists \mu' \in \distr(S) \qdot t \weakstep{a}_R \mu' \en \mu \equiv_R \mu' \]
By Proposition~\ref{propweakdef} it follows that for all $(s,t) \in R$ it holds that
\[s \stepext{a}{\mu} \text{ implies } \exists \mu' \in \distr(S) \qdot t \weakstep{a}_R \mu' \en \mu \equiv_R \mu' \]
Let $R' \subseteq S \times S$ be an equivalence relation such that $R' \supseteq R$. Then, by Lemma~\ref{lemmasuperset} it also holds that
\[s \stepext{a}{\mu} \text{ implies } \exists \mu' \in \distr(S) \qdot t \weakstep{a}_{R'} \mu' \en \mu \equiv_R \mu'\]
Using Proposition 5.2.1.1 and 5.2.1.5 from \cite{Sto02phd} we obtain that
\[s \stepext{a}{\mu} \text{ implies } \exists \mu' \in \distr(S) \qdot t \weakstep{a}_{R'} \mu' \en \mu \equiv_{R'} \mu'\]
Now, applying Proposition~\ref{propweakdef} again, this lemma follows.\qed
\end{proof}

\begin{propositionParam}{\ref{eqrel}}
\eqrel
\end{propositionParam}
\begin{proof}
Reflexivity of $\bb$ trivially holds; the identity relation $\{(s,s) \mid s\in S\}$ can be used as the branching probabilistic bisimulation.

For symmetry, assume that $p \bb q$. Then, there must exist a branching bisimulation $R \subseteq S \times S$ such that $(p,q)\in S$. As every branching bisimulation is an equivalence relation, also $(q,p) \in S$, so also $q \bb p$.

For transitivity, let $p \bb q$ and $q \bb r$. Then, using Proposition~\ref{propweakdef}, there exists an equivalence relation $R_1 \subseteq S \times S$ such that $(p, q) \in R_1$, and for all $(s,t) \in R_1$ it holds that
\[s \weakstep{a}_{R_1} \mu \text{ implies } \exists \mu' \in \distr(S) \qdot t \weakstep{a}_{R_1} \mu' \en \mu \equiv_{R_1} \mu'. \]
Similarly, there exists an equivalence relation $R_2 \subseteq S \times S$ such that $(q, r) \in R_2$, and for all $(s,t) \in R_2$ it holds that
\[s \weakstep{a}_{R_2} \mu \text{ implies } \exists \mu' \in \distr(S) \qdot t \weakstep{a}_{R_2} \mu' \en \mu \equiv_{R_2} \mu'. \]
We define $R_3 = (R_2 \comp R_1) \union (R_1 \comp R_2)$, and let $R$ be the transitive closure of~$R_3$. 
We first prove that $R$ is an equivalence relation by showing (1)~reflexivity, (2)~symmetry, and (3)~transitivity. 
\begin{itemize}
\item[(1)] As $R_1$ are $R_2$ are equivalence relations, they are reflexive; thus, for every state $s \in S$ it holds that $(s,s) \in R_1$ and $(s,s) \in R_2$. Therefore, $(s,s) \in R_2 \comp R_1$ and thus $(s,s) \in R$. 
\item[2)] First observe that when $(x,z) \in R_2 \circ R_1$, then there must be a $y \in S$ such that $(x,y) \in R_1$ and $(y,z) \in R_2$, and therefore by symmetry of $R_1$ and $R_2$ also $(y,x) \in R_1$ and $(z,y) \in R_2$, and thus $(z,x) \in R_1 \comp R_2$.

Now let $(s,t) \in R$. Then there is an integer $n \geq 2$ such that there exists a sequence of states $s_1, s_2, \dots, s_n$ such that $s_1 = s$ and $s_n = t$, and for all $1 \leq i < n$ it holds that
$(s_i, s_{i+1}) \in (R_2 \comp R_1)$ or $(s_i, s_{i+1}) \in (R_1 \comp R_2)$. By the observation above we can reverse the order of the states, obtaining the sequence $s_n, s_{n-1}, \dots, s_1$ such that still $s_n = t$ and $s_1 = s$, and for all $1 \leq i < n$ it holds that $(s_i, s_{i+1}) \in (R_2 \comp R_1)$ or $(s_i, s_{i+1}) \in (R_1 \comp R_2)$. To be precise, when $(s_i, s_{i+1}) \in (R_2 \comp R_1)$, then $(s_{i+1}, s_i) \in (R_1 \comp R_2)$, and when $(s_i, s_{i+1}) \in (R_1 \comp R_2)$, then $(s_{i+1}, s_i) \in (R_2 \comp R_1)$. The sequence obtained in this way proves that $(t,s) \in R$.
\item[(3)] By definition.
\end{itemize}
We now prove that $p \bb r$ by showing that $(p,r) \in R$, and that for all $(s,u) \in R$ it holds that
 \[s \stepext{a}{\mu} \text{ implies } \exists \mu' \in \distr(S) \qdot u \weakstep{a}_R \mu' \en \mu \equiv_R \mu' \]
As $(p,q) \in R_1$ and $(q,r) \in R_2$, it follows immediately that $(p,r) \in R_2 \comp R_1$ and therefore indeed $(p,r) \in R$.

We prove the second part with induction to the number of transitive steps needed to include $(s,u)$ in $R$.

\begin{description}
\item[Base case.]
Let $(s,u) \in R$ because $(s,u) \in R_2 \comp R_1$ (the case where $(s,u) \in R$ because $(s,u) \in R_1 \comp R_2$ can be proven symmetrically).  This implies that there exists a state~$t$ such that $(s,t) \in R_1$ and $(t,u) \in R_2$. Let $s \stepext{a}{\mu}$. Then we know that there exists a $\mu' \in \distr(S)$ such that $t \weakstep{a}_{R_1} \mu'$ and $\mu \equiv_{R_1} \mu'$. By Lemma~\ref{lemmasuperset} it follows that $t \weakstep{a}_R \mu'$, and using Proposition 5.2.1.1 and 5.2.1.5 from \cite{Sto02phd} we see that $\mu \equiv_R \mu'$. 

As $R \supseteq R_2$, we know by Lemma~\ref{lemmasuper2} that for all $(t,u) \in R_2$ it holds that
\[t \weakstep{a}_R \mu' \text{ implies } \exists \mu'' \in \distr(S) \qdot u \weakstep{a}_R \mu'' \en \mu' \equiv_R \mu''. \]

We thus showed that $s \stepext{a}{\mu}$ implies $t \weakstep{a}_{R} \mu'$ (with $\mu \equiv_R \mu')$, and that $t \weakstep{a}_{R} \mu'$ implies $u \weakstep{a}_R \mu''$ (with $\mu' \equiv_R \mu''$). 
Therefore, it follows that if $s \stepext{a}{\mu}$, indeed there exists a $\mu'' \in \distr(S)$ such that $u \weakstep{a}_R \mu''$. As~$\equiv_R$ is an equivalence relation, $\mu \equiv_R \mu''$ follows by transitivity.

\item[Induction hypothesis.] Let $(s,t) \in R$ by $k$ transitive steps. Then,  
\[s \stepext{a}{\mu} \text{ implies } \exists \mu' \in \distr(S) \qdot t \weakstep{a}_R \mu' \en \mu \equiv_R \mu'. \]
\item[Inductive step.]
Let $(s,u) \in R$ by $k+1$ transitive steps. That is, there exists some $t$ such that $(s,t) \in R$ by means of $k$ transitive steps, and either $(t,u) \in R_2 \comp R_1$ or $(t,u) \in R_1 \comp R_2$. We then need to show that
 \[s \stepext{a}{\mu} \text{ implies } \exists \mu'' \in \distr(S) \qdot u \weakstep{a}_R \mu'' \en \mu \equiv_R \mu''. \]

By the induction hypothesis we already know that $s \stepext{a}{\mu}$ implies $t \weakstep{a}_R \mu'$ for some $\mu' \equiv_R \mu$. Moreover, using Proposition~\ref{propweakdef} and the same reasoning as for the base case, we know that $t \weakstep{a}_R \mu'$ implies that $u \weakstep{a}_R \mu''$ for some $\mu'' \equiv_R \mu$. Therefore, by transitivity of $\equiv_R$ the statement holds.
\qed
\end{description}

\end{proof}

\subsection{Proof of Proposition~\ref{prop:strongimpliesweak}}
\begin{lemma}\label{lemmaincl}
Let $\pa$ be a PA, $c \subseteq \{(s,a,\mu) \in \trans \mid a = \tau, \mu \text{ is deterministic}\}$ a set of weakly probabilistically confluent $\tau$-transitions, and $R = \{(s,s') \mid s \join{\ctau}{s'}\}$.  Then, $\mu \ctauequiv \nu$ implies $\mu \equiv_R \mu$.
\end{lemma}
\begin{proof}
Let $\pa$ be a PA, $c \subseteq \{(s,a,\mu) \in \trans \mid a = \tau, \mu \text{ is deterministic}\}$ a set of $\tau$-transitions. Moreover, assume that $\smash{\mu \ctauequiv \nu}$.

Thus, denoting $\nu$  by $\nu = \{t_1 \mapsto p_1, t_2 \mapsto p_2, \dots\}$, there exists a partition $\range{\mu} = \biguplus_{i=1}^n S_i$ such that $n = |\range{\nu}|$ and $\forall 1 \leq i \leq n \colon \mu(S_i) = \nu(t_i) \en$ $\forall s \in S_i \colon s \stepext{\ctau} t_i$.

Now let $R'$ be the smallest equivalence relation that relates the states of every set $S_i$ to each other and to their corresponding $t_i$. That is, for every $S_i$ and for all $s,s' \in S_i$ it holds that $(s,s') \in R'$ and $(s,t_i) \in R'$.

By definition $R'$ is an equivalence relation, and clearly $\mu \equiv_{R'} \nu$. After all, for every $t_i$ it holds that
\begin{align*}
 \nu([t_i]_{R'}) &= \nu(\{t_j \mid (t_i, t_j) \in R' \}) =  \sum_{\substack{1 \leq j \leq n\\(t_i, t_j) \in R'}} \nu(t_j) = \sum_{\substack{1 \leq j \leq n\\(t_i, t_j) \in R'}} \mu(S_j)\\
  &=  \sum_{\substack{1 \leq j \leq n\\\forall s \in S_j \qdot (t_i,s) \in R'}} \mu(S_j) =  \mu([t_i]_{R'})
\end{align*}

Now let $R = \{(s,s') \mid s \join{\ctau}{s'}\}$. A simple tiling argument shows that $R$ is transitive, and reflexivity and symmetry are trivial. Therefore, $R$ is an equivalence relation.
Moreover, as \mbox{$s \stepext{\ctau} t_i$} implies $s \join{\ctau} t_i$, and for every $s,s'\in S_i$ we have $s \join{\ctau} s'$ since they can join at $t_i$, clearly $R$ also relates the states of every set $S_i$ to each other and to their corresponding $t_i$. Since $R'$ is the smallest equivalence relation having this property, it follows that $R \supseteq R'$. Because of this, $\mu \equiv_{R'} \nu$ implies $\mu \equiv_{R} \nu$ (using Proposition 5.2.1.1 and 5.2.1.5 from \cite{Sto02phd}), which is what we wanted to show. \qed
\end{proof}

\begin{propositionParam}{\ref{prop:strongimpliesweak}}
\propstrongconfweak
\end{propositionParam}
\begin{proof}
Let $\pa$ be a PA and $c \subseteq \{(s,a,\mu) \in \trans \mid a = \tau, \mu \text{ is deterministic}\}$ a strongly probabilistically confluent set of $\tau$-transitions. Then, for every transition $s \stepext{\ctau}{t}$ and all $a \in L, \mu \in \distr(S)$ it holds that
\[
     s \stepext{a}{\mu} \implies \left(\left(\exists \nu \in \distr(S) \qdot t \stepext{a}{\nu} \en \smash{\mu  \ctauequiv \nu} \right) \of \left(a = \tau \en \mu = \dirac{t} \right)\right).
\]
Now let $R = \{(s,s') \mid s \join{\ctau}{s'}\}$. As stated in the proof of Lemma~\ref{lemmaincl}, $R$ is an equivalence relation. 

We need to proof that for every path $s \stepextd{\ctau}{t}$ it holds that
\[
     s \stepext{a}{\mu} \implies \left(\left(\exists \nu \in \distr(S) \qdot t \stepext{a}{\nu}  \en \mu \equiv_R \nu \right) \of \left( a = \tau \en \mu \equiv_{R} \dirac{t} \right) \right).
\]
So, let $s \stepext{\ctau} t_1 \stepext{\ctau}\dots \stepext{\ctau}{t_n}$ be such a path, and assume strong probabilistic confluence. Let $s \stepext{a}{\mu}$, and first assume that $a \neq \tau$. Then, 
by definition of strong probabilistic confluence it must hold that $t_1 \stepext{a} \mu_1$ such that $\mu \ctauequiv \mu_1$, and therefore $t_2 \stepext{a} \mu_2$ such that $\mu_1 \ctauequiv \mu_2$, and so on, until $t_n \stepext{a} \mu_n$ such that $\mu_{n-1} \ctauequiv \mu_n$. By Lemma~\ref{lemmaincl} and transitivity of $\equiv_R$, it follows that $\mu \equiv_R \mu_n$. So, the first disjunct of the formula we needed to prove holds (note that for the empty path $s$ it also holds by reflexivity of $\equiv_R$).

Now assume that $a = \tau$. If $\mu \neq \dirac{t_1}$, then the situation is the same as above. If $\mu = \dirac{t_1}$, then it follows that $\mu \equiv_R \dirac{t_n}$ since $t_1 \stepextd{\ctau} t_n$ and thus $(t_1, t_n) \in R$. So, the second disjunct of the formula we needed to prove holds.

So, in both cases $c$ is probabilistically confluent.

\ 

The fact the probabilistic confluence implies weak probabilistic confluence is immediate from the definition, as the former is a restriction of the latter.
\qed
\end{proof}

\subsection{Proof of Theorem~\ref{theoremconfbranching}}
Before proving Theorem~\ref{theoremconfbranching}, we first provide an important lemma.
\begin{lemma}\label{lemmatransitivity}
Let $\pa = \langle S, s^0, L, \trans \rangle$ be a PA and $s,t \in S$. Moreover, let $c$ be a weakly probabilistically confluent set of $\tau$-transitions of $\pa$. Then,
\[
   s \join{\ctau}{t} \mbox{ if and only if } s \tsrc{\ctau} t.
\]
\end{lemma}
\begin{proof}
Let $s \join{\ctau}{t}$. Then, by definition there must exist a state $s' \in S$ such that $s \stepextd{\ctau}{s'}$ and $t \stepextd{\ctau}{s'}$. Therefore, it immediately follows that $s \tsrc{\ctau} t$.

Now let $s \tsrc{\ctau} t$. Then, there must be path such as $s \stepext{\ctau}{s_0} \stepext{\ctau}{s_1} \bstepext{\ctau}{s_2}\bstepext{\ctau}{s_3}\stepext{\ctau}{s_4}\bstepext{\ctau}{t}$. Potentially (as is the case here) the path contains a fragment of the form $s_i \bstepext{\ctau}{s_{i+1}} \stepext{\ctau}{s_{i+2}}$. Clearly this violates the conditions for the path to show that $s \join{\ctau}{t}$, as the arrows point in the wrong direction. Also, note that a path without such a fragment does prove that $s \join{\ctau}{t}$. Therefore, we will show that in any path that can be used to show $s \tsrc{\ctau} t$ we can eliminate these kind of fragments, obtaining a path that proves $s \join{\ctau}{t}$. When $s_i \bstepext{\ctau}{s_{i+1}} \stepext{\ctau}{s_{i+2}}$, then by definition of weak probabilistic confluence either (1) $s_i = s_{i+2}$, or (2) there exists a state $t$ such that $s_{i+2} \stepextd{\ctau} t$ and $s_i \join{\ctau}{t}$. 

In case (1), the whole fragment can just be reduced to the state $s_i$, indeed eliminating the bad fragment.
In case (2), assume that $s_i \join{\ctau}{t}$ is satisfied by the path $s_i \stepext{\ctau}{t_0} \stepext{\ctau} \dots \stepext{\ctau}{t_n} \stepext{\ctau}{t'} \bstepext{\ctau}{t'_1} \bstepext{\ctau} \dots \bstepext{\ctau}{t'_n} \bstepext{\ctau}{t}$, and that $s_{i+2} \stepextd{\ctau} t$ is satisfied by the path $s_{i+2} \stepext{\ctau} t'_{n+1} \stepext{\ctau} \dots  \stepext{\ctau}{t'_m} \stepext{\ctau}{t}$.
Then, the whole fragment can be reduced to $s_i \stepext{\ctau}{t_0} \stepext{\ctau} \dots \stepext{\ctau}{t_n} \stepext{\ctau}{t'} \bstepext{\ctau}{t'_1} \bstepext{\ctau} \dots \bstepext{\ctau}{t'_n} \bstepext{\ctau}{t}\bstepext{\ctau}{t'_m} \bstepext{\ctau} \dots \bstepext{\ctau}{t'_{n+1}} \bstepext{\ctau}{s_{i+2}}$, which is of the correct form.

Repeating this for all bad fragments, a path proving $s \join{\ctau}{t}$ appears. \qed
\end{proof}

\begin{theoremParam}{\ref{theoremconfbranching}}
\theoremconfbranching
\end{theoremParam}
\begin{proof}
Let $\pa = \langle S, s^0, L, \trans \rangle$ be a PA and $c$ a weakly probabilistically confluent set of $\tau$-transitions. We prove that $s \join{\ctau}{s'}$ implies $s \bb s'$. Clearly, when this hold also $s \stepext{\ctau}{s'}$ implies that $s \bb s'$. Then, as $\bb$ is an equivalence relation, the theorem follows.

Let $s,s' \in S$ such that $s \join{\ctau}{s'}$. To prove that indeed $s \bb s'$, we show that $R = \{(s,t) \mid s \join{\ctau}{t}\}$ is a branching probabilistic bisimulation. Obviously $(s,s') \in R$. Lemma~\ref{lemmatransitivity} and the fact that $\tsrc{\ctau}$ is an equivalence relation imply that $R$ is an equivalence relation.

To show that $R$ is a branching probabilistic bisimulation, let $(s,t) \in R$ be an arbitrary pair of states in $R$. We prove that 
\[s \stepext{a}{\mu} \text{ implies } \exists \mu' \in \distr(S) \qdot t \weakstep{a}_R \mu' \en \mu \equiv_R \mu'.\]
Let $u$ be the joining state of $s$ and $t$, i.e., $s \stepextd{\ctau}{u}$ and $t \stepextd{\ctau}{u}$. Let $s \stepext{a}{\mu}$. We make a case distinction based on whether $a \neq \tau$ or $a = \tau$.
\begin{itemize}
\item Assume that $a \neq \tau$.  From the definition of weak probabilistic bisimulation it immediately follows that 
\[
 \exists u' \in S \qdot u \stepextd{\ctau}{u'} \en \exists \nu \in \distr(S) \qdot u' \stepext{a}{\nu} \en \mu \equiv_{R} \nu 
 \]
Now let $\sched$ be a scheduler choosing with probability $1$ the transitions from $t$ corresponding to the path $t \stepextd{\ctau}{u'}$. Then, let $\sched$ choose $u' \stepext{a}{\nu}$ with probability $1$, followed by $\bottom$ with probability $1$. Clearly the final state distribution of $\sched$ is $\nu$, and indeed $\mu \equiv_R \nu$ by definition of weak probabilistic confluence. Moreover, as the scheduler only follows $\ctau$-transitions before it selects $u' \stepext{a}{\nu}$, the branching condition is satisfied.

\item Assume that $a = \tau$. From the definition of weak probabilistic bisimulation it then follows that
\[
   \exists u' \in S \qdot u \stepextd{\ctau}{u'} \en \left(\left(   \exists \nu \in \distr(S) \qdot u' \stepext{a}{\nu} \en \mu \equiv_{R} \nu \right) {} \of {} \left( \mu \equiv_{R} \dirac{u'} \right)\right) 
\]
When the first disjunct is satisfied the above reasoning applies, so from now on assume that $\exists u' \in S \qdot u \stepextd{\ctau}{u'} \en\mu \equiv_{R} \dirac{u'}$. If $t=u'$, then $t \weakstep{a}_R \mu$ is satisfied by the first clause of the definition of branching probabilistic bisimulation as $a = \tau$ and $\mu \equiv_R \dirac{u'} = \dirac{t}$. If $t \neq u'$, then the transition can be mimicked by the scheduler choosing with probability $1$ the transitions from~$t$ corresponding to the path $t \stepextd{\ctau}{u'}$ and then choosing $\bottom$ with probability~$1$. Clearly the final state distribution of $\sched$ is $\dirac{u'}$, and indeed $\mu \equiv_R \dirac{u'}$ by the assumption we made. Moreover, as the scheduler only follows $\ctau$-transitions, the branching condition is satisfied.\qed
\end{itemize} 
\end{proof}

\subsection{Proof of Proposition~\ref{propunion}}
\begin{propositionParam}{\ref{propunion}}
\propunion
\end{propositionParam}
\begin{proof}
Let $\pa$ be a PA and $c \subseteq \{(s,a,\mu) \in \trans \mid a = \tau, \mu \text{ is deterministic}\}$ a weakly probabilistically confluent set of $\tau$-transitions. Then, for every path $s \stepextd{\ctau}{t}$ and all $a \in L, \mu \in \distr(S)$ it holds that
\begin{align*}
     s \stepext{a}{\mu} \implies & \exists t' \in S \qdot t \stepextd{\ctau}{t'} \en \\
  & \left(\left(   \exists \nu \in \distr(S) \qdot t' \stepext{a}{\nu} \en \mu \equiv_{R} \nu \right) {} \of {} \left( a = \tau \en \mu \equiv_{R} \dirac{t'} \right)\right)
\end{align*}
where $R = \{(s,s') \mid s \join{\ctau}{s'}\}$ (and $R$ is an equivalence relation).

Let $c'$ be a different weakly probabilistically confluent set of $\tau$-transitions. Then, for every path $s \stepextd{\tau_{c'}}{t}$ and all $a \in L, \mu \in \distr(S)$ it holds that
\begin{align*}
     s \stepext{a}{\mu} \implies & \exists t' \in S \qdot t \stepextd{\tau_{c'}}{t'} \en \\
  & \left(\left(   \exists \nu \in \distr(S) \qdot t' \stepext{a}{\nu} \en \mu \equiv_{R'} \nu \right) {} \of {} \left( a = \tau \en \mu \equiv_{R'} \dirac{t'} \right)\right)
\end{align*}
where $R' = \{(s,s') \mid s \join{\tau_{c'}}{s'}\}$ (and $R$ is an equivalence relation).

Now, taking $c'' = c \union c'$, for every path $s \stepextd{\tau_{c''}}{t}$ and all $a \in L, \mu \in \distr(S)$ it should hold that
\begin{align*}
     s \stepext{a}{\mu} \implies & \exists t' \in S \qdot t \stepextd{\tau_{c''}}{t'} \en \\
  & \left(\left(   \exists \nu \in \distr(S) \qdot t' \stepext{a}{\nu} \en \mu \equiv_{R''} \nu \right) {} \of {} \left( a = \tau \en \mu \equiv_{R''} \dirac{t'} \right)\right)
\end{align*}
where $R'' = \{(s,s') \mid s \join{\tau_{c''}}{s'}\}$. The fact that $R''$ is again an equivalence relation follows easily from the restrictions on the $\tau_{c}$- and $\tau_{c'}$-steps.
Moreover, the required implication follows from a common tiling argument.

A similar argument can be given for probabilistically confluent sets, and for strongly probabilistically confluent sets.\qed
\end{proof}

\subsection{Proof of Theorem~\ref{theoremeqclasses}}
\begin{theoremParam}{\ref{theoremeqclasses}}
\theoremeqclasses
\end{theoremParam}
\begin{proof}
Theorem~\ref{theoremconfbranching} already showed that all states that can reach each other via $\ctau$-transitions are branching probabilistically bisimilar. It is well known that such states can therefore be merged, preserving branching probabilistic bisimulation.

The equivalence relation $R$ needed to show this relates the states of the disjoint union of $\pa$ and $\left(\pa / {\tsrc{\ctau}}\right)$ in such a way that every equivalence class contains precisely one of the states $[s]_{\tsrc{\ctau}}$ of $\left(\pa / {\tsrc{\ctau}}\right)$ and all the states $s'$ in $\pa$ such that $s' \in [s]_{\tsrc{\ctau}}$.
Clearly, $(s^0, [s^0]_{\tsrc{\ctau}}) \in R$. In the remainder of the proof we omit the subscript of equivalence classes $[s]_{\tsrc{\ctau}}$, as they are always the same.

Now, let $\{s_1, s_2, \dots, s_n, [s_1]\}$ be one of the equivalence classes of $R$. Because of Theorem~\ref{theoremconfbranching} all the states $s_1, s_2, \dots, s_n$ are branching probabilistically bisimilar, so we only still need to show that $[s_1]$ can mimic the behaviour of $s_1, s_2, \dots, s_n$ and vice versa. 

So, assume that for instance $s_2 \stepext{a} \mu$. Then, by definition $[s_1] \stepext{a} \nu$ such that $\nu([t]) = \sum_{t' \in [t]} \mu(t')$ for every $[t] \in S/{\tsrc{\ctau}}$. This implies that \mbox{$\mu \equiv_R \nu$} by the construction of $R$.

Conversely, let $[s_1] \stepext{a} \mu$. Then, by definition there must exist a state \mbox{$s_i \in [s]$} such that $s_i \stepext{a}{\nu}$ and $\forall {[t] \in S/{\tsrc{\ctau}}} \qdot \mu([t]) = \sum_{t' \in [t]} \nu(t')$. So, $\mu \equiv_R \nu$. Therefore, $s_i$ can mimic the behaviour of $[s_1]$. As all the states $s_1, s_2, \dots, s_n$ are branching probabilistically bisimilar, they can all mimic each other, so therefore all states can mimic the behaviour of $[s_1]$.
\qed
\end{proof}

\subsection{Proof of Theorem~\ref{theoremrep}}

\begin{theoremParam}{\ref{theoremrep}}
\theoremrep
\end{theoremParam}
\begin{proof}
Theorem~\ref{theoremeqclasses} already showed that $\left(\pa / {\tsrc{\ctau}}\right) \bb \pa$ if $c$ is weakly probabilistically confluent, and by Proposition~\ref{prop:strongimpliesweak} this also holds for probabilistically confluent sets $c$. As $\bb$ is an equivalence relation by Proposition~\ref{eqrel}, it therefore suffices to prove that $\left(\pa / {\tsrc{\ctau}}\right) \bb \left(\pa / \phi_c\right)$. In this proof we will omit the subscript of equivalence classes $[s]_{\tsrc{\ctau}}$, as they are always the same.

By definition every state of $\pa / {\tsrc{\ctau}}$ is an equivalence class $[s]$, and every state of $\pa / \phi_c$ is a representative $\phi_c(s)$. We define the relation $R$ to be the reflexive and symmetric closure of 
\[
\{([s], \phi_c(s)) \mid s \in S\}, 
\]
and show that it is a branching probabilistic bisimulation. Clearly $R$ is an equivalence relation, and by definition $([s^0], \phi_c(s^0)) \in R$. To show that $R$ is a branching probabilistic bisimulation we prove that $[s] \stepext{a}{\mu}$ implies that there exists a $\mu'$ such that $\phi_c(s) \weakstep{a} \mu'$ and $\mu \equiv_R \mu'$, and that  $\phi_c(s) \stepext{a}{\mu}$ implies that there exists a $\mu'$ such that $[s] \weakstep{a} \mu'$ and $\mu \equiv_R \mu'$.

\begin{itemize}
\item
Let $[s] \stepext{a}{\mu}$. We prove that the exists a $\mu'$ such that $\phi_c(s) \weakstep{a} \mu'$ and $\mu \equiv_R \mu'$ by showing that, assuming $\mu([t]) = p$ for an arbitrary state $t$, there exists a $\mu'$ such that $\phi_c(s) \weakstep{a} \mu'$ and \mbox{$\mu'(\phi_c(t)) = p$}.

By Definition~\ref{defeqrel}, there must exist a state $s' \in [s]$ in $\pa$ and a $\mu''$ such that $s' \stepext{a}{\mu''}$ and $\forall {[t] \in S/{\tsrc{\ctau}}} \qdot \mu([t]) = \sum_{t' \in [t]} \mu''(t')$. That is, $\mu''$ also assigns probability $p$ to the event of going to a state in the equivalence class~$[t]$.

Now, by definition of representatives $s' \stepextd{\ctau}{\phi_c(s)}$, and therefore, by definition of probabilistic confluence it must be the case that 
\[\left(\exists \nu \in \distr(S) \qdot \phi_c(s) \stepext{a}{\nu}  \en \mu'' \equiv_{R'} \nu \right) \of \left( a = \tau \en \mu'' \equiv_{R'} \dirac{\phi_c(s)} \right),
\]
where $R' = \{(s,s') \mid s \join{\ctau}{s'}\}$. First assume that $a \neq \tau$. Then, in $\pa$ we have $\phi_c(s) \stepext{a} \nu$ with $\nu \equiv_{R'} \mu''$. Given the definition of $R'$ and Lemma~\ref{lemmatransitivity}, this implies that $\nu([t]) = \mu''([t]) = p$. As $[t]$ is exactly the set of all states that have $\phi_c(t)$ as their representative, by Definition~\ref{defmodulorep} we also have $\phi_c(s) \stepext{a}{\nu'}$ with $\nu'(\phi_c(t)) = p$ in $\pa / \phi_c$. As the existence of a transition implies the existence of a weak step, this finishes this part of the proof.

When $a = \tau$, either the above holds, or $\mu'' \equiv_{R'} \dirac{\phi_c(s)}$. In the latter case, as we also already knew that $\mu''$ assigns probability $p$ to the event of going to a state in the equivalence class~$[t]$, apparently $\phi_c(s) \in [t]$ and $p=1$. From $\phi_c(s) \in [t]$ it follows by definition that $\phi_c(s) = \phi_c(t)$. By definition of branching steps $\phi_c(s) \weakstep{\tau} \dirac{\phi_c(s)}$, and given the above indeed $\dirac{\phi_c(s)}(\phi_c(t)) = p$.
\item Let $\phi_c(s) \stepext{a}{\mu}$, and let $\mu(\phi_c(t)) = p$ for some state $t$. We prove that there exists a $\mu'$ such that $[s] \weakstep{a} \mu'$ and \mbox{$\mu'([t]) = p$}.
As $\phi_c(s) \stepext{a}{\mu}$, there must exist a transition $t' \stepext{a}{\mu'}$ in the original PA such that $\phi_c(t') = \phi_c(s)$ and
\[ \forall s' \in \phi_c(S) \qdot \mu(s') = \mu'(\{s'' \in S \mid \phi_c(s'') = s'\}).\]
So, because we assumed $\mu(\phi_c(t)) = p$, it should hold that $\mu'(\{s'' \in S \mid \phi_c(s'') = \phi_c(t)\} = p$. Stated otherwise, and recognising that the set of states with the same representative as $t$ is precisely the set $[t]$, we get $\mu'([t]) = p$. As $t' \stepext{a}{\mu'}$ such that $\mu'([t]) = p$, and because $s$ and $t'$ have the same representative, also $[s] \stepext{a}{\mu''}$ such that $\mu''([t]) = p$. Observing that a normal step implies a weak step, we're done.\qed
\end{itemize}
\end{proof}

\subsection{Proof of Theorem~\ref{theoremsymbolic}}

\begin{theoremParam}{\ref{theoremsymbolic}}
\theoremsymbolic
\end{theoremParam}
\begin{proof}
Let $X$ be an LPPE and $\pa$ its underlying PA. So, by the operational semantics the state set $S$ of this PA contains precisely all vectors $\vec{g} \in \vec{G}$.

Let $i$ be a summand such that $\forall \vec{g} \in \vec{G}, \vec{d_i} \in \vec{D_i} \qdot a_i(\vec{g}, \vec{d_i}) = \tau \en \exists \vec{e_i} \in \vec{E_i} \qdot f_i(\vec{g}, \vec{d_i}, \vec{e_i}) = 1$ and
 for every summand $j$ it holds that
 \begin{align}\label{confluenceformule2}
 \begin{array}{ll}
 \smash{\big(}c_i(\vec{g},\vec{d_i}) \en c_j(\vec{g},\vec{d_j})\smash{\big)} \rightarrow &\smash{\big(} i = j \en \vec{n_i}(\vec{g}, \vec{d_i}) = \vec{n_j}(\vec{g}, \vec{d_j})\smash{\big)} \of \\[5pt]
& \!\!\!\!\!\left(\begin{array}{ll}
  & c_j(\vec{n_i}(\vec{g},\vec{d_i}), \vec{d_j}) \en c_i(\vec{n_j}(\vec{g},\vec{d_j},\vec{e_j}), \vec{d_i})\\
\en & a_j(\vec{g},\vec{d_j}) = a_j(\vec{n_i}(\vec{g},\vec{d_i}),\vec{d_j})\\
\en & f_j(\vec{g}, \vec{d_j}, \vec{e_j}) = f_j(\vec{n_i}(\vec{g}, \vec{d_i}), \vec{d_j}, \vec{e_j})\\
\en & \vec{n_j}(\vec{n_i}(\vec{g},\vec{d_i}),\vec{d_j},\vec{e_j}) =\vec{n_i}(\vec{n_j}(\vec{g},\vec{d_j},\vec{e_j}),\vec{d_i})
\end{array}\right) \end{array}
\end{align} 
By the operational semantics, the transitions generated by $i$ are those transitions $\vec{g} \stepext{a} \mu$ such that there is a choice of local variables $\vec{d_i} \in \vec{D_i}$ such that
\[
    c_i(\vec{g}, \vec{d_i}) \en a_i(\vec{g}, \vec{d_i}) = a \en \forall \vec{e_i} \in \vec{E_i} \qdot \mu(\vec{n_i}(\vec{g}, \vec{d_i}, \vec{e_i})) =\!\!\!\!\!\!\!\!\!\!\!\!\!\!\!\!\!\sum_{\substack{\smash{\vec{e'_i}} \in \vec{E_i}\\\vec{n_i}(\vec{g}, \vec{d_i}, \vec{e_i}) = \vec{n_i}(\vec{g}, \vec{d_i}, \vec{e_i'})}}\!\!\!\!\!\!\!\!\!\!\!\!\!\!\!\!\!f_i(\vec{g}, \vec{d_i}, \vec{e_i'}).
\]
Let $c$ be the set containing these transitions. We prove that $c$ is probabilistically confluent by showing that it is strongly probabilistically confluent, relying on Proposition~\ref{prop:strongimpliesweak}. Note that the sets $c_i$ from all confluent summands $i$ can be combined into a single confluent set by Proposition~\ref{propunion}.

Let $\vec{g} \stepext{a} \mu$ be an arbitrary transition in $c$, and let $\vec{d'_i} \in \vec{D_i}$ be the local variables that had to be chosen for $i$ to generate it.~So,
\begin{align*}
    c_i(\vec{g}, \vec{d'_i}) \en a_i(\vec{g}, \vec{d'_i}) = a \en \forall \vec{e_i} \in \vec{E_i} \qdot \mu(\vec{n_i}(\vec{g}, \vec{d'_i}, \vec{e_i})) =\!\!\!\!\!\!\!\!\!\!\!\!\!\!\!\!\!\sum_{\substack{\smash{\vec{e'_i}} \in \vec{E_i}\\\vec{n_i}(\vec{g}, \vec{d'_i}, \vec{e_i}) = \vec{n_i}(\vec{g}, \vec{d'_i}, \vec{e_i'})}}\!\!\!\!\!\!\!\!\!\!\!\!\!\!\!\!\!f_i(\vec{g}, \vec{d'_i}, \vec{e_i'}).
\end{align*}
Because \mbox{$\forall \vec{g} \in \vec{G}, \vec{d_i} \in \vec{D_i} \qdot a_i(\vec{g}, \vec{d_i}) = \tau \en \exists \vec{e_i} \in \vec{E_i} \qdot f_i(\vec{g}, \vec{d_i}, \vec{e_i}) = 1$}, it follows that $a = \tau$. Moreover, $\mu$ is deterministic (as it assigns probability $1$ to the next state determined by $\vec{n_i}(\vec{g}, \vec{d'_i}, \vec{e_i})$, where $\vec{e_i}$ is the unique element of $\vec{E_i}$ such that $f_i(\vec{g}, \vec{d'_i}, \vec{e_i}) = 1$). We use $\vec{g'}$ to denote this unique target state. Thus, using the notation $\vec{n_i}(\vec{g}, \vec{d'_i})$ for the unique target state given the global state $\vec{g}$ and local variables $\vec{d'_i}$, we have $\vec{g'} = \vec{n_i}(\vec{g}, \vec{d'_i})$.

So, we indeed can write $\vec{g} \stepext{a} \mu$ as $\vec{g} \stepext{\ctau} \vec{g'}$. To show that $c$ is strongly probabilistically confluent it remains to show that for every transition $\vec{g} \stepext{a}{\mu}$ 
\begin{equation}\label{eqdisjucts}
     \left(\exists \nu \in \distr(S) \qdot \vec{g'} \stepext{a}{\nu} \en \smash{\mu \ctauequiv \nu} \right) \of \left(a = \tau \en \mu = \dirac{\vec{g'}} \right).
\end{equation}

Let $\vec{g} \stepext{a}{\mu}$ be such a transition for which this needs to be shown. Let $j$ be the summand from which it originates, and let $\vec{d'_j}$ be the local variables that had to be chosen for $j$ to generate it. So, 
\begin{equation*}\label{eq2}
    c_j(\vec{g}, \vec{d'_j}) \en a_j(\vec{g}, \vec{d'_j}) = a \en \forall \vec{e_j} \in \vec{E_j} \qdot \mu(\vec{n_j}(\vec{g}, \vec{d'_j}, \vec{e_j})) =\!\!\!\!\!\!\!\!\!\!\!\!\!\!\!\!\!\!\!\!\!\sum_{\substack{\smash{\vec{e'_j}} \in \vec{E_j}\\\vec{n_j}(\vec{g}, \vec{d'_j}, \vec{e_j}) = \vec{n_j}(\vec{g}, \vec{d'_j}, \vec{e_j'})}}\!\!\!\!\!\!\!\!\!\!\!\!\!\!\!\!\!\!\!\!\!f_j(\vec{g}, \vec{d'_j}, \vec{e_j'})
\end{equation*}

Now, as we showed above that $c_i(\vec{g}, \vec{d_i'})$ and $c_j(\vec{g}, \vec{d'_j})$ hold, we can apply Equation~(\ref{confluenceformule2}). Therefore, we know that either $\smash{\big(} i = j \en \vec{n_i}(\vec{g}, \vec{d'_i}) = \vec{n_j}(\vec{g}, \vec{d'_j})\smash{\big)}$, or for every $\vec{e_j}\in \vec{E_j}$ it holds that
\[
\left(\begin{array}{ll}
  & c_j(\vec{g'}, \vec{d'_j}) \en c_i(\vec{n_j}(\vec{g},\vec{d'_j},\vec{e_j}), \vec{d'_i})\\
\en & a_j(\vec{g},\vec{d'_j}) = a_j(\vec{g'},\vec{d'_j})\\
\en & f_j(\vec{g}, \vec{d'_j}, \vec{e_j}) = f_j(\vec{g'},\vec{d'_j}, \vec{e_j})\\
\en & \vec{n_j}(\vec{g'},\vec{d'_j},\vec{e_j}) =\vec{n_i}(\vec{n_j}(\vec{g},\vec{d'_j},\vec{e_j}),\vec{d'_i})
\end{array}\right)
\] 
(where we already substituted $\vec{g'}$ for $\vec{n_i}(\vec{g},\vec{d'_i})$.

In the first case, both $\vec{g} \stepext{\ctau} \vec{g'}$ and $\vec{g} \stepext{a}{\mu}$ result from the same summand. Therefore, $a = \tau$ is immediate, and indeed $\mu = \dirac{\vec{g'}}$ because $\vec{g'} = \vec{n_i}(\vec{g}, \vec{d'_i})$, $\vec{n_i}(\vec{g}, \vec{d'_i}) = \vec{n_j}(\vec{g}, \vec{d'_j})$, and $\vec{n_j}(\vec{g}, \vec{d'_j})$
 denotes the unique target state of $j$ given the local variables $\vec{d'_j}$. Therefore, the second disjunct of Equation~\ref{eqdisjucts} holds.
 
In the second case, we know several things. First of all, $c_j(\vec{g'}, \vec{d'_j})$, so~$j$ is still enabled using the local variables $\vec{d'_j}$ in state $\vec{g'}$. Moreover, $a_j(\vec{g},\vec{d'_j}) = a_j(\vec{g'},\vec{d'_j})$, so since above we already showed that $a_j(\vec{g}, \vec{d'_j}) = a$, it follows that taking $j$ from~$\vec{g'}$ using the local variables $\vec{d'_j}$ also results in the action $a$. So, this shows that there indeed exists a distribution $\nu \in \distr(S)$ such that $\vec{g'} \stepext{a} \nu$.

The final part of the proofs explains that $\smash{\mu \ctauequiv \nu}$. For this, we need to show that there exists a partition $\range{\mu} = \biguplus_{k=1}^n S_k$ such that $n = |\range{\nu}|$ and $\forall 1 \leq k \leq n \colon \mu(S_k) = \nu(t_k) \en \forall s \in S_k \colon s \stepext{\ctau} t_k$. 

To see this, assume that $\nu = \{\vec{g'_1} \mapsto p_1, \vec{g'_2} \mapsto p_2, \dots\}$. Now, let the partition be given by $S_k = \{ \vec{n_j}(\vec{g},\vec{d'_j},\vec{e'_j}) \mid \vec{e_j'} \in \vec{E_j},  \vec{n_i}(\vec{n_j}(\vec{g},\vec{d'_j},\vec{e'_j}),\vec{d'_i}) = \vec{g'_k} \}$. Clearly, by construction the partition contains precisely as many elements as $\range{\nu}$. It is also indeed a valid partition: (1) there is no state in $\range{\mu}$ that is not present in at least one of the $S_k$'s, since every such state can be written as $\vec{n_j}(\vec{g},\vec{d'_j},\vec{e'_j})$ and because of the requirement that $c_i(\vec{n_j}(\vec{g},\vec{d'_j},\vec{e_j}), \vec{d'_i})$ each of them goes to at least one of the $\vec{g_k'}$'s (which is also guaranteed to be in the support of $\nu$ if it is in the support of $\mu$, as follows from the computation below); (2) there is no state in $\range{\mu}$ that occurs in more than one $S_k$, since any such state can be written as $\vec{n_j}(\vec{g},\vec{d'_j},\vec{e'_j})$ and therefore precisely occurs only in the $S_k$ such that $\vec{n_i}(\vec{n_j}(\vec{g},\vec{d'_j},\vec{e'_j}),\vec{d'_i}) = \vec{g'_k}$.

Furthermore, for any $S_k$:
\begin{align*}
    \mu(S_k) &= \mu(\{ \vec{n_j}(\vec{g},\vec{d'_j},\vec{e'_j}) \mid \vec{e_j'} \in \vec{E_j},  \vec{n_i}(\vec{n_j}(\vec{g},\vec{d'_j},\vec{e'_j}),\vec{d'_i}) = \vec{g'_k} \})\\
   &= \sum_{\substack{\smash{\vec{e'_j}} \in \vec{E_j}\\\vec{n_i}(\vec{n_j}(\vec{g},\vec{d'_j},\vec{e'_j}),\vec{d'_i}) = \vec{g'_k}}} f_j(\vec{g}, \vec{d_j'}, \vec{e_j'})\\
   &= \sum_{\substack{\smash{\vec{e'_j}} \in \vec{E_j}\\\vec{n_i}(\vec{n_j}(\vec{g},\vec{d'_j},\vec{e'_j}),\vec{d'_i}) = \vec{g'_k}}} f_j(\vec{g'}, \vec{d_j'}, \vec{e_j'})\displaybreak[4]\\
   &=\,\quad\sum_{\substack{\smash{\vec{e'_j}} \in \vec{E_j}\\\vec{n_j}(\vec{g'},\vec{d_j'}, \vec{e_j}) = \vec{g'_k}}}\quad\,f_j(\vec{g'}, \vec{d_j'}, \vec{e_j'})\\
   &= \nu(\vec{g_k'})
\end{align*}
The first equality just unfolds the construction of $S_k$, the second applies the operational semantics of summands, the third applied the requirement that $f_j(\vec{g}, \vec{d'_j}, \vec{e_j}) = f_j(\vec{g'},\vec{d'_j}, \vec{e_j})$ for every $\vec{e_j} \in \vec{E_j}$, the fourth applies the requirement that $\vec{n_j}(\vec{g'},\vec{d'_j},\vec{e_j}) =\vec{n_i}(\vec{n_j}(\vec{g},\vec{d'_j},\vec{e_j}),\vec{d'_i})$ for every $\vec{e_j} \in \vec{E_j}$, and the last equality again uses the operational semantics.

The fact that every state $s \in S_k$ has a $\ctau$-transition to $\vec{g_k'}$ follows directly 
from the requirement that $c_i(\vec{n_j}(\vec{g},\vec{d'_j},\vec{e_j}), \vec{d'_i})$ for every $\vec{e_j} \in \vec{E_j}$ and the construction that $\vec{n_i}(\vec{n_j}(\vec{g},\vec{d'_j},\vec{e'_j}),\vec{d'_i}) = \vec{g_k'}$.
\qed
\end{proof}

\section{Case study: a leader election protocol}\label{appcase}
The LPPE for the leader election protocol discussed in Section~\ref{seccasestudy} is shown in Figure~\ref{LPPE}. For readability, in every summand we only show the parameters that are used or updated. 
As summations can use an existing parameter name for a local variable, statements such as $\emph{d\_2} := \emph{d\_2}$ occur. This means that, in the next state, the global variable \emph{d\_2} will have the value of the local variable~\emph{d\_2}. We use the notation $\text{reset}(x)$ to denote that a variable $x$ is reset to its initial~value.

\begin{figure*}[b!]
\newcommand{\opvul}{{}+{}}
\centering\scalebox{0.95}{\centering\framebox{
\begin{minipage}{0.7\linewidth}\smaller
\vspace{-0.3cm}\begin{align*}
\hspace{0cm}Z(\text{\it val\_1} &: \textit{\{1..6\}}, \textit{set\_1} : \textit{Bool}, \textit{pc\_2} : \textit{\{1..4\}}, \textit{d\_2} : \textit{\{1..6\}}, \textit{e\_2} : \textit{\{1..6\}},\\
& \!\!\!\!\!\!\!\!\!\!\!\!\!\!\textit{val\_3} : \textit{\{1..6\}}, \textit{set\_3} : \textit{Bool}, \textit{pc\_4} : \textit{\{1..4\}}, \textit{d\_4} : \textit{\{1..6\}}, \textit{e\_4} : \textit{\{1..6\}}) =\\ 
& \textit{pc\_2}=1 \Rightarrow \tau\!\! \psum_{\textit{d\_2}: \textit{\{1..6\}}} \tfrac{1}{6} \colon Z(\textit{pc\_2} := 2,\textit{d\_2} := \textit{d\_2}, \text{reset}(\textit{e\_2})) & (1)&\\ 
\opvul & \textit{pc\_4}=1 \Rightarrow \tau\!\! \psum_{\textit{d\_4} : \textit{\{1..6\}}} \tfrac{1}{6} \colon Z(\textit{pc\_4} := 2,\textit{d\_4} := \textit{d\_4},\text{reset}(\textit{e\_4}))&(2)\\
\opvul &  \textit{pc\_2}=2 \en \neg\textit{set\_3} \Rightarrow \tau \!\psum_{k: \{1\}} 1.0 \colon Z(\textit{pc\_2}:=3, \text{reset}(\textit{e\_2}),\textit{val\_3} := \textit{d\_2}, \textit{set\_3} := \textit{true}) &(3)\\
\opvul & \textit{pc\_4}=2 \en \neg\textit{set\_1} \Rightarrow \tau\! \psum_{k: \{1\}} 1.0 \colon Z(\textit{val\_1} := \textit{d\_4},\textit{set\_1} := \textit{true},\textit{pc\_4}:= 3,\text{reset}(\textit{e\_4}))&(4)\\
\opvul & \textit{pc\_2}=3 \en \textit{set\_1} \Rightarrow\tau\! \psum_{k: \{1\}} 1.0 \colon Z(\textit{set\_1} := \textit{false},\textit{pc\_2} := 4,\textit{e\_2} := \textit{val\_1}) &(5)\\ 
\opvul & \textit{pc\_4}=3 \en \textit{set\_3} \Rightarrow\tau\! \psum_{k: \{1\}} 1.0 \colon Z(\textit{set\_3} := \textit{false},\textit{pc\_4} := 4,\textit{e\_4} := \textit{val\_3})&(6)\\ 
\opvul & \textit{pc\_2}=4 \en \textit{d\_2}=\textit{e\_2} \Rightarrow\tau\!\! \psum_{\textit{d\_2}: \textit{\{1..6\}}} \tfrac{1}{6} \colon Z(\textit{pc\_2} := 2,\textit{d\_2} := \textit{d\_2},\text{reset}(\textit{e\_2})) &(7)\\ 
\opvul & \textit{pc\_4}=4 \en \textit{d\_4}=\textit{e\_4} \Rightarrow\tau\!\! \psum_{\textit{d\_4}: \textit{\{1..6\}}} \tfrac{1}{6} \colon Z(\textit{pc\_4} := 2,\textit{d\_4} := \textit{d\_4},\text{reset}(\textit{e\_4})) &(8)\\ 
\opvul & \textit{pc\_2}=4 \en \textit{d\_2} > \textit{e\_2} \Rightarrow \textit{leader}(\textit{one}) \psum_{k: \{1\}} 1.0 \colon Z(\textit{pc\_2} := 1,\text{reset}(\textit{d\_2}),\text{reset}(\textit{e\_2})) &(9)\\ 
\opvul & \textit{pc\_4}=4 \en \textit{d\_4} > \textit{e\_4} \Rightarrow \textit{leader}(\textit{two}) \psum_{k: \{1\}} 1.0 \colon Z(\textit{pc\_4} := 1,\text{reset}(\textit{d\_4}),\text{reset}(\textit{e\_4})) &(10)\\ 
\opvul & \textit{pc\_2}=4 \en \textit{d\_2}<\textit{e\_2} \Rightarrow \textit{follower}(\textit{one}) \psum_{k: \{1\}} 1.0 \colon Z(\textit{pc\_2}:=1,\text{reset}(\textit{d\_2}),\text{reset}(\textit{e\_2})) &(11)\\ 
\opvul & \textit{pc\_4}=4 \en \textit{d\_4}<\textit{e\_4} \Rightarrow \textit{follower}(\textit{two}) \psum_{k: \{1\}} 1.0 \colon Z(\textit{pc\_4} := 1,\text{reset}(\textit{d\_4}),\text{reset}(\textit{e\_4}))&(12)
 \end{align*}\larger
\end{minipage}
}}
\caption{The LPPE of a leader election protocol.}
\label{LPPE}
\end{figure*}

We will now give an idea of how confluence detection works for this LPPE. In fact, for the case studies we implemented a tool that does this kind of reasoning automatically.

As the summands $9$, $10$, $11$, and $12$ do not have a $\tau$-action, and the target states of the summands $1$, $2$, $7$, and $8$ are not determined by a deterministic distribution, the only candidates for confluent summands are $3$, $4$, $5$, and 6. It turns out that all of these are indeed confluent. To show for instance the confluence of summand~3, we need to prove that, for every summand $j$, and all $\vec{g}, \vec{d_i}, \vec{d_j}$, and $\vec{e_j}$:
\begin{align*}
 \begin{array}{ll}
 \smash{\big(}c_i(\vec{g},\vec{d_i}) \en c_j(\vec{g},\vec{d_j})\smash{\big)} \rightarrow &\smash{\big(} i = j \en \vec{n_i}(\vec{g}, \vec{d_i}) = \vec{n_j}(\vec{g}, \vec{d_j})\smash{\big)} \of \\[5pt]
& \!\!\!\!\!\left(\begin{array}{ll}
  & c_j(\vec{n_i}(\vec{g},\vec{d_i}), \vec{d_j}) \en c_i(\vec{n_j}(\vec{g},\vec{d_j},\vec{e_j}), \vec{d_i})\\
\en & a_j(\vec{g},\vec{d_j}) = a_j(\vec{n_i}(\vec{g},\vec{d_i}),\vec{d_j})\\
\en & f_j(\vec{g}, \vec{d_j}, \vec{e_j}) = f_j(\vec{n_i}(\vec{g}, \vec{d_i}), \vec{d_j}, \vec{e_j})\\
\en & \vec{n_j}(\vec{n_i}(\vec{g},\vec{d_i}),\vec{d_j},\vec{e_j}) =\vec{n_i}(\vec{n_j}(\vec{g},\vec{d_j},\vec{e_j}),\vec{d_i})
\end{array}\right) \end{array}
\end{align*} 
Note that summands $1$, $5$, $7$, $9$, and $11$ can never be enabled at the same time as summand $3$, as summand $3$ requires $\textit{pc\_2} = 2$ and these summands all require $\textit{pc\_2}$ to have another value. Also, $6$ can never be enabled at the same time as $3$ because of their contradictory requirements on $\textit{set\_3}$. Therefore, the formula we need to prove holds vacuously for all these summands (as the left-hand side of the implication can never be true).

So, we only still need to prove the formula for the summands $2$, $4$, $8$, $10$, and $12$, and for $3$ itself.  Simple observation shows that summand $3$ only deals with the variables $\textit{pc\_2}, \textit{set\_3}, \textit{e\_2}, \textit{val\_3}$, and $\textit{d\_2}$, and that none of the summands $2$, $4$, $8$, $10$, and $12$ either uses or updates any of these variables. This immediately implies that summand $3$ cannot disable any of these summands or the other way around, that summand $3$ cannot influence the actions or probabilities of these summands, and that the state that is reached after first executing summand $3$ and then one of these summands must be identical to the state that is reached when doing this the other way around (given the same probabilistic choice). Therefore, the formula holds for all possible valuations of its parameters (by means of the second disjunct).

To see that summand $3$ is confluent with itself, notice that it does not have any local variables to choose from. Therefore, the first disjunct of the formula holds trivially.

The same kind of reasoning can be applied to show that summands $4$, $5$, and $6$ are confluent.
\fi
\end{document}